\documentclass[11pt, a4paper]{article}

\usepackage{authblk}
\usepackage[dvipsnames]{xcolor}
\usepackage[margin=1in,bottom=1in]{geometry}

\usepackage{amsmath}
\usepackage{amsthm}
\usepackage{amssymb}
\usepackage{amsfonts}
\usepackage{mathtools}
\allowdisplaybreaks[1]

\usepackage{enumerate}
\usepackage{enumitem}
\usepackage{needspace}
\usepackage{float}
\usepackage[title]{appendix}

\usepackage{adjustbox}
\usepackage[nodisplayskipstretch]{setspace}

\usepackage[most]{tcolorbox}

\usepackage{indentfirst}
\usepackage{comment}

\usepackage{booktabs}
\usepackage{array}
\usepackage{colortbl}
\usepackage{subfigure}
\usepackage{makecell}

\usepackage[english]{babel}
\usepackage[T1]{fontenc}
\AtBeginDocument{%
  \DeclareFontShape{T1}{cmr}{m}{scit}{<->ssub*cmr/m/sc}{}%
}

\usepackage[algoruled]{algorithm2e}
\SetAlgorithmName{Protocol}{Protocol}{List of Protocols}

\usepackage{tikz}

\usepackage{xurl} 
\usepackage[pagebackref,colorlinks,linkcolor=NavyBlue,anchorcolor=blue,citecolor=NavyBlue,hypertexnames=false]{hyperref}

\hypersetup{breaklinks=true}

\usepackage[capitalize,noabbrev]{cleveref}
\crefname{ineq}{inequality}{inequalities}
\creflabelformat{ineq}{#2{\upshape(#1)}#3}
\crefname{fact}{fact}{facts}
\crefname{equation}{equation}{equations}
\crefname{algorithm}{protocol}{protocols} 
\crefname{remark}{remark}{remarks}
\crefname{conjecture}{conjecture}{conjectures}
\crefname{problem}{problem}{problems}

\usepackage{thmtools}
\makeatletter
\@ifundefined{newcounteralias}{}{%
  \renewcommand\thmt@autorefsetup{%
    \@xa\def\csname\thmt@envname autorefname
    \@xa\endcsname\@xa{\thmt@thmname}%
  }%
}
\makeatother

\usepackage{thm-restate}
\declaretheorem[style=plain,numberwithin=section]{theorem}
\declaretheorem[style=plain,numberlike=theorem]{lemma,corollary}

\declaretheorem[style=plain,numberlike=theorem]{definition}
\declaretheorem[style=plain,numberlike=theorem]{proposition}
\declaretheorem[style=plain,numberlike=proposition]{fact}

\declaretheorem[style=definition,numberlike=theorem]{problem,conjecture}
\numberwithin{equation}{section}

\makeatletter
\def\@buildmath#1{%
  \expandafter\def\csname bb#1\endcsname{\ensuremath{\mathbb{#1}}}%
  \expandafter\def\csname bf#1\endcsname{\ensuremath{\mathbf{#1}}}%
  \expandafter\def\csname sf#1\endcsname{\ensuremath{\mathsf{#1}}}%
  \expandafter\def\csname cal#1\endcsname{\ensuremath{\mathcal{#1}}}%
  \expandafter\def\csname rm#1\endcsname{\ensuremath{\mathrm{#1}}}%
  \expandafter\def\csname tt#1\endcsname{\ensuremath{\mathtt{#1}}}%
}
\def\@buildmathletters#1{%
  \ifx#1\relax\else
    \@buildmath{#1}%
    \expandafter\@buildmathletters
  \fi
} 
\@buildmathletters ABCDEFGHIJKLMNOPQRSTUVWXYZabcdefghijklmnopqrstuvwxyz\relax
\makeatother

\newcommand{\parheading}[1]{%
  \par\addvspace{1em}%
  \noindent\emph{#1}\enspace\ignorespaces%
}

\newcommand{\EXP}{\textnormal{\textsf{EXP}}\xspace}
\newcommand{\NEXP}{\textnormal{\textsf{NEXP}}\xspace}
\newcommand{\BPP}{\textnormal{\textsf{BPP}}\xspace}

\newcommand{\QCMA}{\textnormal{\textsf{QCMA}}\xspace}
\newcommand{\QCMAone}{\textnormal{\textsf{QCMA}\textsubscript{1}}\xspace}
\newcommand{\QMA}{\textnormal{\textsf{QMA}}\xspace}
\newcommand{\PreciseQMA}{\textnormal{\textsf{PreciseQMA}}\xspace}
\newcommand{\IP}{\textnormal{\textsf{IP}}\xspace}
\newcommand{\QIP}{\textnormal{\textsf{QIP}}\xspace}
\newcommand{\QIPtwo}{\textnormal{\textsf{QIP(2)}}\xspace}
\newcommand{\QIPtwoone}{\textnormal{\textsf{QIP(2)}\textsubscript{1}}\xspace}
\newcommand{\QIPLtwo}{\textnormal{\textsf{QIPL(2)}}\xspace}

\newcommand{\qqQAM}{\textnormal{\textrm{qq}\text{-}\textsf{QAM}}\xspace}
\newcommand{\qqQAMone}{\textnormal{\textrm{qq}\text{-}\textsf{QAM}\textsubscript{1}}\xspace}

\newcommand{\QAM}{\textnormal{\textsf{QAM}}\xspace}
\newcommand{\QCAM}{\textnormal{\textsf{QCAM}}\xspace}

\newcommand{\QMAone}{\textnormal{\textsf{QMA}\textsubscript{1}}\xspace}
\newcommand{\QMAtwo}{\textnormal{\textsf{QMA}(2)}\xspace}

\newcommand{\MA}{\textnormal{\textsf{MA}}\xspace}
\newcommand{\AM}{\textnormal{\textsf{AM}}\xspace}

\newcommand{\StoqMA}{\textnormal{\textsf{StoqMA}}\xspace}

\newcommand{\protocol}[2]{{#1}\!\rightleftharpoons\!{#2}}

\newcommand{\PSPACE}{\textnormal{\textsf{PSPACE}}\xspace}
\newcommand{\QMAL}{\textnormal{\textsf{QMAL}}\xspace}

\newcommand{\FixedCI}{\textnormal{\textsc{FixedCI}}\xspace}

\DeclarePairedDelimiter\rbra{\lparen}{\rparen}
\DeclarePairedDelimiter\sbra{\lbrack}{\rbrack}
\DeclarePairedDelimiter\cbra{\{}{\}}
\DeclarePairedDelimiter\abs{\lvert}{\rvert}
\DeclarePairedDelimiter\norm{\lVert}{\rVert}

\let\ket\relax\DeclarePairedDelimiter\ket{\lvert}{\rangle}
\let\bra\relax\DeclarePairedDelimiter\bra{\langle}{\rvert}

\newcommand{\ketbra}[2]{\ensuremath{\ket{#1}\!\bra{#2}}}
\renewcommand{\bra}[1]{\langle #1 \rvert}
\renewcommand{\ket}[1]{\lvert #1 \rangle}
\newcommand{\braket}[3]{\langle #1 | #2 | #3 \rangle}
\newcommand{\innerprod}[2]{\langle #1 | #2 \rangle}

\newcommand{\Tr}{\mathrm{Tr}}

\newcommand{\F}{\mathrm{F}}

\newcommand{\spanset}{\mathrm{span}}

\newcommand{\yes}{\mathrm{yes}}
\newcommand{\no}{\mathrm{no}}

\newcommand{\EPR}{\mathrm{EPR}}
\newcommand{\init}{\mathrm{init}}
\newcommand{\midd}{\mathrm{mid}}
\newcommand{\final}{\mathrm{final}}
\newcommand{\inward}{\mathrm{inward}}
\newcommand{\qq}{\mathrm{qq}}
\newcommand{\In}{\mathrm{in}}
\newcommand{\Out}{\mathrm{out}}
\newcommand{\match}{\mathrm{match}}

\newcommand{\Toffoli}{\textnormal{\textsc{Toffoli}}\xspace}
\newcommand{\Had}{\textnormal{\textsc{H}}\xspace}
\newcommand{\CNOT}{\textnormal{\textsc{CNOT}}\xspace}

\renewcommand{\Pr}[1]{\mathrm{Pr}\!\left[#1\right]}

\newcommand{\binset}{\{0,1\}}

\newcommand{\poly} {\operatorname{poly}}

\DeclareMathOperator\diag{diag}

\begin{document}

\setlength{\abovedisplayskip}{6pt}
\setlength{\belowdisplayskip}{6pt}

\title{Achieving perfect completeness for one- and two-message quantum proof systems}

\author[1]{Yupan Liu\thanks{Email: \url{yupan.liu@epfl.ch}}}
\author[1]{Thomas Vidick\thanks{Email: \url{thomas.vidick@epfl.ch}}}
\affil[1]{%
  School of Computer and Communication Sciences,\protect\\
  \'Ecole Polytechnique F\'ed\'erale de Lausanne%
}

\date{}
\maketitle

\begin{abstract}
While quantum interactive proof systems using at least \emph{three} messages can achieve perfect completeness, as shown by \hyperlink{cite.KW00}{Kitaev and Watrous~(STOC 2000)}, whether perfect completeness is achievable for \emph{one}- and \emph{two}-message quantum proof systems has remained open. For the one-message case, whether $\sf QMA$ can achieve perfect completeness was posed as an open problem in \hyperlink{cite.Watrous00}{Watrous~(FOCS 2000)} and \hyperlink{cite.AN02}{Aharonov and Naveh~(2002)};  for the two-message case, the corresponding problems were (implicitly) posed in \hyperlink{cite.JUW09}{Jain, Upadhyay, and Watrous~(FOCS 2009)} and \hyperlink{cite.KLGN19}{Kobayashi, Le Gall, and Nishimura~(SICOMP, 2019)}. 

In this work, we establish that ${\sf QIP}(2)$, ${\rm qq}\text{-}{\sf QAM}$, $\sf QAM$, and $\sf QMA$ can achieve perfect completeness. Here ${\rm qq}\text{-}{\sf QAM}$ denotes the class of promise problems admitting two-message \emph{quantum-public-coin} quantum interactive proof systems in which the verifier's only message consists of half-EPR pairs. Our main technical contributions are as follows: 
\begin{enumerate}[label={\upshape(\arabic*)}]
    \item For $\sf QMA$ (and directly for $\sf QAM$), an \emph{exactly constructible} block-encoded matrix whose kernel certifies \emph{yes} instances, constructed from the acceptance operator induced by the verification circuit. This also yields an alternative proof of ${\sf QCMA}={\sf QCMA}_1$.
    \item For ${\sf QIP}(2)$ (and implicitly ${\rm qq}\text{-}{\sf QAM}$), a new turn-halving transformation that preserves completeness and ensures that the resulting proof system retains at least \emph{two} messages, provided that the terminal state before the final measurement is \emph{efficiently preparable}. 
\end{enumerate}
\end{abstract}


\section{Introduction}

The investigation of quantum interactive proof systems was initiated in~\cite{Wat03}, fifteen years after the notion of \emph{interactive proof systems} was introduced in the classical setting~\cite{Babai85,GMR85}. In a quantum interactive proof system for a promise problem $(\calI_\yes,\calI_\no)$, a quantum polynomial-time \emph{verifier} interacts with a computationally unbounded but untrusted \emph{prover}. While these two parties share no entanglement initially, they may create entanglement during the interaction, which consists of at most polynomially many exchanged messages. Given an input $x\in\calI_\yes \cup \calI_\no$, the prover claims that $x\in\calI_\yes$, after which the verifier initiates an interactive protocol and decides whether to ``accept'' or ``reject'' the claim at the end. The protocol has completeness $c$, meaning that if $x\in\calI_{\yes}$, there exists a prover strategy under which the verifier accepts with probability at least $c$. We say that a protocol achieves \emph{perfect completeness} if the verifier accepts with certainty ($c=1$). The protocol also has soundness $s$, meaning that if $x\in \calI_\no$ then the verifier accepts with probability at most $s$ regardless of the prover's strategy. 

While quantum interactive proof systems are \emph{no more powerful} than their classical counterparts when the promise gap $c-s$ is at least inverse-polynomial, namely $\QIP=\PSPACE$~\cite{JJUW11} and $\IP=\PSPACE$~\cite{LFKN92,Shamir92}, quantum interactive proof systems exhibit several distinctive properties and subtleties. Perhaps the most striking is the \emph{parallelization}: any quantum interactive proof system can be parallelized to use only \emph{three} messages~\cite{KW00}. In comparison, \emph{constant}-message classical interactive proof systems can be parallelized to \emph{two} messages~\cite{Babai85,GS86}, whereas parallelizing classical interactive proof systems with a \emph{polynomial} number of messages to a constant number of messages would collapse the polynomial-time hierarchy to the second level~\cite{BHZ87}.

In addition to parallelization, the notion of \emph{public coins} becomes more subtle in the quantum setting. Classically, public coins mean that the verifier sends only uniformly random questions (i.e., its random coins), and every classical interactive proof system can be transformed into a public-coin protocol by adding two messages~\cite{GS86}. All \emph{three}-message quantum interactive proof systems can also be made public-coin~\cite{MW05}, and this idea was later extended to a turn-halving transformation~\cite{KKMV09}. Recursively applying this transformation yields an alternative three-message parallelization that simultaneously produces a public-coin protocol. However, for two-message quantum interactive proof systems, the notion of public coins is no longer unique: the verifier can prepare EPR pairs and send exactly one half of each pair, which can be viewed as \emph{quantum public coins}; the resulting model is denoted by \qqQAM{}. A classification theorem for constant-message public-coin quantum interactive proof systems was subsequently established in~\cite{KLGN19}, showing that such models reduce either to \PSPACE or to one of three two-message classes: \qqQAM, \QAM{}, and \QCAM{} (all messages are classical), with the verifier sending classical public coins in the latter two models.  

\paragraph{Achieving perfect completeness in quantum proof systems.}
Beyond parallelization and public coins, another important property is perfect completeness, which can always be achieved in classical proof systems~\cite{FGMSZ89}. The transformation of~\cite{KW00} achieves perfect
completeness by adding \emph{two} messages, while the subsequent transformation of~\cite{KLGN15} achieves this property by adding only \emph{one} message. Since the aforementioned parallelization results preserve perfect completeness, quantum interactive proof systems with at least \emph{three} messages can always achieve perfect completeness. 

For one-message proof systems, in which the only message is referred to as a \emph{witness}, it is well-known that $\MA=\MA_1$~\cite{ZF87,GZ11} in the classical setting, where the subscript $1$ denotes perfect completeness and (at most) constant soundness error. 
In the quantum setting, whether \QMA{} can achieve perfect completeness has been listed as an open problem since the early 2000s~\cite{Watrous00,AN02}. Indeed, all quantumly relativizing black-box techniques for achieving perfect completeness were ruled out in~\cite{Aaronson09}. 
When the witness is \emph{classical}, $\QCMA=\QCMAone$ was established in~\cite{JKNN12}, by extending to \QCMA{} an idea from a special case of \QMA{}~\cite{NWZ09}: using an appropriate \emph{dyadic} parameter to encode the acceptance probability.
Subsequently, a variant of \QMA{} in which the verifier initially shares a constant number of EPR pairs with the prover was shown to achieve perfect completeness~\cite{KLGN15}, while remaining no more powerful than \QMA{}~\cite{BSW11}. More recently, it was shown in~\cite{JW25} that the completeness error for \QMA{} can be made \emph{doubly exponentially small}, a bound that is \emph{optimal} among all relativizing black-box techniques~\cite{AHW25}. 

The least understood regime has been that of two-message quantum proof systems. It was established in~\cite{KLGN19} that $\QCMA=\QCMA_1$ implies that $\QCAM=\QCAM_1$, and the same approach can be extended to show $\QAM \subseteq \qqQAMone$. However, whether \QIPtwo{} or its public-coin variants \qqQAM{} and \QAM{} can achieve perfect completeness remains open, and these questions were (implicitly) posed as open problems in~\cite{JUW09,KLGN19}. 

\subsection{Main results}
In this work, we establish that the following four classes can achieve perfect completeness: 

\begin{restatable}{theorem}{qipPerfectCompleteness}
    \label{thm:QIP(2)-informal}
    $\QIPtwo=\QIPtwoone$.
\end{restatable}

\begin{restatable}{theorem}{qqPerfectCompleteness}
    \label{thm:qqQAM-informal}
    $\qqQAM=\qqQAMone$.
\end{restatable}

More specifically, applying our technique to a \QIPtwo{} proof system with completeness at least $2/3$ and soundness at most $1/3$ yields a new $\QIPtwo_1$ proof system with soundness at most 
\[ s_\star = \frac{1}{2}\rbra[\bigg]{1 + \frac{1+\sqrt{2}}{\sqrt{6}}} < 0.993.\] 
Our transformation for \qqQAM{} achieves the same soundness threshold. By further combining our transformations with the corresponding parallel-repetition procedures in~\cite[Theorem 6]{KW00} and~\cite[Lemma 3.4]{KLGN19}, respectively, the soundness error of the resulting two-message quantum proof systems can be reduced to at most $1/3$. 

\vspace{1em}
While our results for \QIPtwo{} and \qqQAM{} hold for \emph{any} reasonable gateset, achieving perfect completeness for \QMA{} and \QAM{} is \emph{gateset-dependent}:

\begin{restatable}{theorem}{qmaPerfectCompleteness}
    \label{thm:QMA-informal}
    $\QMA^{\calG} = \QMA^{\calG}_1$.
\end{restatable}

Here, we use the fixed gateset $\calG\coloneqq\cbra{\Had,X,\Toffoli,S}$, where $S\coloneqq\diag(1,i)$.\footnote{$\QMA_1^\calG$ is equivalent to the number-field-dependent $\QMA_1^{\calG_4}$ over $\bbQ(i)$, and the latter further equals $\QMA_1^{\calG_2}$ by~\cite[Theorem~1.1]{Rudolph26}, with $\calG_2 = \cbra{X,  \CNOT, \Toffoli, \Had\otimes\Had}$.} Applying our technique to a $\QMA^\calG$ proof system with completeness at least $3/4$ and soundness at most $1/4$ yields a new $\QMA^\calG_1$ proof system with soundness at most $1-2^{-16}$. 

Following the lifting argument from \QCMA{} to \QCAM{} in~\cite[Theorem 1.6]{KLGN19}, one can eliminate the completeness error caused by the verifier's random challenges (i.e., classical public coins) using the textbook technique proving $\BPP \subseteq \Sigma_2 \cap \Pi_2$~\cite{Lautemann83,Sipser83}, thereby extending perfect completeness from \QMA{} to \QAM{}:

\begin{restatable}{corollary}{qaaPerfectCompleteness}
    \label{thm:QAM-informal}
    $\QAM^{\calG} = \QAM^{\calG}_1$.
\end{restatable}

\subsection{Proof overview: Achieving perfect completeness for \QIPtwo{} and \qqQAM{}}
\label{subsec:proof-overview}

The proof of \Cref{thm:QIP(2)-informal} begins with the following two known elementary transformations and then applies our new turn-halving transformation.

\begin{description}
    \item[Exact-half calibration.] The first transformation is the \emph{single-qubit exact-half calibration} underlying the notion of \emph{perfect rewindability} for quantum interactive proof systems introduced in~\cite{KKMV09,KLGN15}. For two-message quantum proof systems, this transformation adds a qubit to the prover's response.\footnote{This transformation applies to any quantum proof system with at least two messages in which the last message is sent by the prover, while preserving the number of messages.} For \emph{yes} instances with acceptance probability $p \geq 2/3$, the honest prover is expected to initialize the additional qubit as $\ket{1}$ with probability $1/(2p)$, and $\ket{0}$ with the remaining probability. The verifier accepts only if both the new qubit is measured to be $1$ and the original test accepts.\footnote{The reason that the creation of the qubit is left to the prover is that in general the verifier does not know $p$, which is prover-dependent.} Consequently, the resulting \emph{two-message} proof system has completeness exactly $1/2$. For \emph{no} instances, the additional test can only decrease the acceptance probability, even if the new qubit is entangled with the rest of the prover's message. 

    \item[EPR-pair completion.] The second transformation specializes the technique for achieving perfect completeness in~\cite{KW00} to a proof system with completeness exactly $1/2$. In particular, the verifier $V$ executes the calibrated proof system \emph{without} performing the final measurement and coherently records the decision in a single-qubit register $\sfZ$. Next, $V$ sends to the prover all qubits in the message and its private register that purify $\sfZ$, receives a single-qubit register $\sfZ'$, and projects the state on $(\sfZ, \sfZ')$ onto an EPR pair. For \emph{yes} instances, the terminal state on $(\sfZ, \sfZ')$ is \emph{exactly} an EPR pair, so $V$ accepts with certainty, achieving perfect completeness. For \emph{no} instances, if the original proof system has soundness $s < 1/2$, Uhlmann's theorem~\cite{Uhlmann76} implies that the resulting proof system has soundness at most $g_\EPR(s)\coloneqq \frac{1}{2}+\sqrt{s(1-s)}$, which equals $\frac{1}{2}+\frac{\sqrt{2}}{3}$ when $s = 1/3$. 
\end{description}

\subsubsection{An endpoint-inward turn-halving transformation}
\label{subsubsec:inward-turn-halving}

The starting observation is that the terminal state produced by EPR completion is exactly the EPR pair $\ket{\Phi^+}_{(\sfZ, \sfZ')} \coloneqq (\ket{00}+\ket{11})/\sqrt{2}$ for \emph{yes} instances. This \emph{efficiently preparable} terminal state enables our new \emph{endpoint-inward} turn-halving transformation, in which the verifier starts coherently from the initial state $\ket{\psi_\init}$ and the terminal state $\ket{\psi_\final}$ in the message and its private registers, while the prover must make the two branches agree at the midpoint state $\rho_\midd$. By contrast, the known turn-halving transformation is \emph{middle-outward}~\cite{KKMV09}: the verifier tosses a public coin to choose whether to check the interaction forward or backward from the prover-supplied midpoint state $\rho_\midd$. These two orientations are illustrated in \Cref{fig:midpoint-orientations}.

\begin{figure}[!ht]
\centering
\begin{tikzpicture}[
  >=stealth,
  state/.style={draw,rounded corners,align=center,inner sep=5pt,
                text width=3.1cm},
  kkmv/.style={->,thick,draw=black!70},
  folded/.style={->,very thick,draw=NavyBlue},
  every node/.style={font=\normalsize}
]
  \node[state] (init) at (-5.1,0)
    {Fixed initial endpoint $\ket{\psi_\init}$};
  \node[state] (mid) at (0,0)
    {One common midpoint $\rho_\midd$};
  \node[state] (term) at (5.1,0)
    {Fixed final endpoint $\ket{\psi_\final}$};
  \draw[kkmv,bend right=22] (mid) to
    node[above,align=center] {\cite{KKMV09}: $b=1$\\simulate \emph{backward}} (init);
  \draw[kkmv,bend left=22] (mid) to
    node[above,align=center] {\cite{KKMV09}: $b=0$\\simulate \emph{forward}} (term);
  \draw[folded,bend right=22] (init) to
    node[below,align=center] {\underline{This work}: first branch\\\emph{forward}} (mid);
  \draw[folded,bend left=22] (term) to
    node[below,align=center] {\underline{This work}: second branch\\\emph{backward}} (mid);
\end{tikzpicture}
\caption{Turn-halving transformations with different orientations.}
\label{fig:midpoint-orientations}
\end{figure}
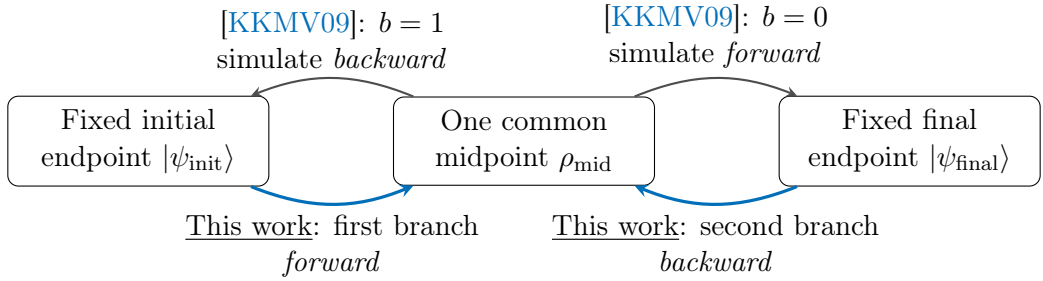

\paragraph{Inward turn-halving transformation.}

Starting from a \emph{four}-message quantum proof system obtained after calibration and completion, the verifier prepares a coherent superposition of $\ket{\psi_\init}$ and $\ket{\psi_\final}$, with the two branches labelled by single-qubit registers $\sfB$ and $\sfB'$, where the prover does not touch $\sfB$ but receives the qubit in $\sfB'$. The prover executes the $\ket{\psi_\init}$ branch forward through the first half of the original interaction and the $\ket{\psi_\final}$ branch backward through the second half. Perfect completeness guarantees that, for \emph{yes} instances, the two branches reach the same midpoint state, after which $(\sfB, \sfB')$ contains an EPR pair. For \emph{no} instances, the two unnormalized branch states
have norm at most $1$ and overlap of magnitude at most
$\sqrt{s}$, implying that the resulting proof system has
soundness at most
\[ g_\inward(s) \coloneqq \frac{1+\sqrt{s}}{2}.\] 
Consequently, the transformation folds four messages into \emph{two} while preserving perfect completeness, with resulting soundness $s_\star \coloneqq g_\inward(g_\EPR(1/3)) < 1$. 

\paragraph{Achieving perfect completeness in \texorpdfstring{\qqQAM}{qq-QAM}.}
The exact-half calibration and EPR-pair completion also apply to \qqQAM proof systems, but the inward turn-halving transformation \emph{no longer} admits the quantum-public-coin form: although both forward and backward branches begin with maximally mixed states, these states have \emph{different dimensions} whenever the verifier sends more than one \emph{quantum} public coin (i.e., half of an EPR pair). 

We instead encode both endpoint conditions into a \emph{fixed-target} instance of \textsf{Close Image} from~\cite{Wat02a,HMW14}. Let $\sfW$ denote the verifier's private register containing the retained halves of the original EPR pairs in the \qqQAM{} proof system, and let the single-qubit register $\sfZ$ contain the calibrated decision state $\calR(\rho)$ for a proposed midpoint $\rho$. A perfect midpoint must satisfy both $\Tr_{\setminus\sfW}(\rho) = I_\sfW / \dim(\sfW)$ and $\calR(\rho)=I_\sfZ/2$, where $\Tr_{\setminus\sfW}(\rho)$ denotes tracing out all registers of $\rho$ except $\sfW$. We then mimic the inward turn-halving transformation by defining a quantum channel that pads the missing qubits at either endpoint, thereby aligning their dimensions, with $\sfF$ denoting a single-qubit flag register:\footnote{A similar two-boundary construction also appears in the \qqQAM{}-hardness of the \textsc{Close Image to Totally Mixed} problem in~\cite[Lemma 4.2]{KLGN19}.}
\[ \calN_\qq(\rho) \coloneqq \frac{1}{2} \ketbra{0}{0}_\sfF \otimes \frac{I_\sfZ}{2} \otimes \Tr_{\setminus\sfW}(\rho) + \frac{1}{2} \ketbra{1}{1}_\sfF \otimes \calR(\rho)_\sfZ \otimes \frac{I_\sfW}{\dim(\sfW)}. \]
Consequently, a midpoint satisfying both conditions is mapped
to the fixed target state
\[ \sigma_\qq \coloneqq \frac{I_\sfF}{2} \otimes \frac{I_\sfZ}{2} \otimes \frac{I_\sfW}{\dim(\sfW)}. \]

The resulting proof system follows immediately from the \QIPtwo{} containment of \textsc{Close Image}~\cite{Wat02a,HMW14} and its \qqQAM{} adaptation in~\cite[Lemma 4.1]{KLGN19}: the verifier (i) sends the halves of EPR pairs purifying $\sigma_\qq$, (ii) receives the non-output qubits of the mixed-state quantum circuit implementation $Q_\qq$ of $\calN_\qq$, and (iii) applies $Q_\qq^\dagger$. Exact preimages yield perfect completeness for \emph{yes} instances; Uhlmann's theorem and other basic properties of the fidelity yield the desired soundness bound $s_\star$ for \emph{no} instances.

\subsection{Proof overview: Achieving perfect completeness for \QMA{}}
\label{subsec:qma-proof-overview}

We start from the acceptance operator $M \coloneqq \bra{\bar{0}} V_x^\dagger \ketbra{1}{1}_\Out V_x \ket{\bar{0}}$ of a \QMA{} verification circuit $V_x$ over the gateset $\calG$, where $x\in\calI$ for a promise problem $\calI\in\QMA$, $w$ denotes the number of qubits used by the witness, and $M$ is a $2^w\times 2^w$ positive semidefinite matrix. 
Using error reduction for \QMA{}~\cite[Theorem 3.3]{MW05}, we may assume $\sigma_{\max}(M) = \norm{M} \geq 1-1/(2d)$ for \emph{yes} instances and $M \preceq I/(2d)$ for \emph{no} instances, where $d\coloneqq w+3$ and $\sigma_{\max}(M)$ denotes the largest singular value of $M$. 
To establish \Cref{thm:QMA-informal}, we construct a $2^w\times 2^w$ matrix $K$, specified by polynomial-length classical labels $(z,m)$ and satisfying the following properties: 
\begin{itemize}
    \item For \emph{yes} instances, there exists a valid label $(z,m)$ such that $\ker(K) \neq \cbra{0}$.    
    \item For \emph{no} instances, every valid label $(z,m)$ satisfies $\sigma_{\min}(K) \geq 1/16$.
\end{itemize}
In particular, for \emph{yes} instances, the $\QMA_1$ proof consists of  classical labels $(z,m)$ and a quantum state $\ket{\psi_z}$ in the kernel of $K$. 
At a high level, the construction of $K$ in \Cref{eq:intro-qma-matrix} serves the same purpose as the \emph{shifted-OR trick} in the proof of
$\MA\subseteq\MA_1$~\cite{ZF87,GZ11}, namely to eliminate the remaining completeness error while preserving a soundness gap.

\paragraph{Constructing the matrix $K$ whose kernel certifies \emph{yes} instances.} We begin by defining the matrix $K$ which is specified by the classical labels $(z,m)$:
\begin{equation}
    \label{eq:intro-qma-matrix}
    K\coloneqq\alpha
    \underbrace{(I-M)^d}_{\text{OR-type repetition}}
    \underbrace{\rbra*{I-\beta(I+M)^d\ketbra{z}{z}}}_{\text{rank-one perturbation of }I}.
\end{equation}
Here, $\alpha$ and $\beta$ are positive numbers to be specified later.
The two factors have distinct roles:
\begin{itemize}
    \item \emph{OR-type repetition.} The first factor reduces completeness error. An exact block encoding of $(I-M)^d$,\footnote{See \Cref{def:block-encoding} for a formal definition of exact block encoding.} with all ancillary qubits initialized to $\ket{0}$ and rejection when all ancillary qubits are measured to be zero, yields acceptance probability $1-(1-\lambda)^{2d}$ on a normalized eigenvector of $M$ with eigenvalue $\lambda$, matching the OR of $2d$ independent trials with success probability $\lambda$.\footnote{This polynomial alone gives perfect completeness only when $\lambda=1$.}
    \item \emph{Rank-one perturbation.} For a \emph{yes} instance, to achieve perfect completeness, we choose $z$ maximizing $\braket{z}{(I+M)^d}{z}$ and set $\beta=1/\braket{z}{(I+M)^d}{z}$, so that
    \begin{equation}
        \label{eq:intro-qma-matrix-oblique}
        K = \alpha(I-M)^d (I-\ketbra{\phi_z}{z}), \quad\text{where } \ket{\phi_z} \coloneqq \frac{(I+M)^d\ket{z}}{\bra{z} (I+M)^d \ket{z}}.
    \end{equation}
    Here, $\ket{\phi_z}$ is an \emph{unnormalized} vector. One can verify that $I-\ketbra{\phi_z}{z}$ is an \emph{oblique projection}, meaning that $\rbra*{I-\ketbra{\phi_z}{z}}^2=I-\ketbra{\phi_z}{z}$ (since $\innerprod{z}{\phi_z}=1$) while $I-\ketbra{\phi_z}{z}$ may not be Hermitian. 
    With this choice of $\beta$, for any $\alpha>0$, a direct calculation gives
    \[ \ker(K) \supseteq \ker(I-\ketbra{\phi_z}{z}) = \spanset\cbra*{ \ket{\psi_z} }, \quad\text{where } \ket{\psi_z}\coloneqq\frac{(I+M)^d\ket{z}}{\norm*{ (I+M)^d \ket{z} }}. \]
    We will specify later how the classical certificate $m$ determines $\alpha$ and $\beta$.
\end{itemize}

There are further reasons why directly testing $I-\ketbra{\phi_z}{z}$ in \Cref{eq:intro-qma-matrix-oblique} can be problematic: (i) $\beta$, which equals $1/\braket{z}{(I+M)^d}{z}$, is typically \emph{not dyadic}, so we choose $\alpha$ such that both $\alpha$ and $\alpha\beta$ are dyadic; and (ii) the operator norm of $I-\ketbra{\phi_z}{z}$ can be of order $2^{w/2}$,
\[ \norm*{I-\ketbra{\phi_z}{z}} = \frac{\norm*{ (I+M)^d\ket{z} }}{\braket{z}{(I+M)^d}{z}} = \frac{1}{\abs*{\innerprod{z}{\psi_z}}}, \]
so this operator need not be a \emph{contraction} and may require an exponentially large normalization factor in a block encoding, an issue avoided through left multiplication by $(I-M)^d$.

\paragraph{Implementing the new verification circuit $V'_x$.}
To implement $K$ as a verification circuit, it remains to show that $\alpha$ and $\beta$ \emph{require only polynomial-bit precision}. 
Let $h$ denote the number of Hadamard gates in $V_x$, so that $2^hM$ has Gaussian integer entries. For an integer $m$ represented with $\ell$ bits satisfying $2^{hd+2}\leq m<2^\ell$, where $\ell\coloneqq(h+1)d+1$, let $k$ be the number of bits needed to represent $m$ in binary, excluding leading zeros. Then, $2^{k-1}\leq m<2^k$ and $k\leq\ell$. 

We now choose appropriate $\alpha$ and $\beta$ such that every valid pair of labels satisfies $1/4 \leq \alpha < 1/2$ and $0 < \beta \leq 1/4$. With these choices, the coefficients $\alpha$ and $\alpha\beta$ are dyadic and satisfy
\[
    \alpha+\alpha\beta=\frac{m+2^{hd}}{2^{k+1}}\leq2\alpha<1,
    \quad\text{where } \alpha \coloneqq \frac{m}{2^{k+1}} \text{ and } \beta \coloneqq \frac{2^{hd}}{m}.
\]

We construct an exact block encoding $U_K$ of $K$ over $\calG$ using block encodings of POVM operators and their products~\cite[Lemmas~26 and~30]{GSLW19}, together with exact dyadic linear combinations (\Cref{lemma:uniform-selector}) based on the LCU construction~\cite{BCCKS15,GSLW19} and the dyadic-amplitude construction of~\cite{Rudolph26}. The resulting circuit implementation uses $O(d)$ calls to $V_x,V_x^\dagger$, and polynomially many one- or three-qubit gates from $\calG$ and ancillary qubits.

The new verifier receives the classical labels $(z,m)$ and a $w$-qubit quantum witness $\ket{\psi}$, and proceeds as follows:
\begin{enumerate}[label={\upshape(\roman*)}]
    \item Measure the classical labels and reject unless $2^{hd+2}\leq m<2^\ell$.
    \item Construct an exact block encoding $U_K$ of $K$ by first determining the dyadic coefficients $\alpha$ and $\alpha\beta$.
    \item\label{item:exact-kernel-test} Apply $U_K$ to $\ket{\psi}$ and the ancillary qubits of $U_K$ initialized to $\ket{\bar{0}}$, and reject if measuring these ancillary qubits yields the all-zero outcome.
\end{enumerate}

It is noteworthy that Step \ref{item:exact-kernel-test} follows the \emph{exact kernel test} in~\cite[Lemma 3.8]{Rudolph26}, which was originally applied to an exact block encoding of a local Hamiltonian (obtained via its Pauli decomposition). In our construction $V'_x$, we instead use an exact block encoding of $K$.

\paragraph{Analysis of $V'_x$.}
For fixed valid labels and any quantum state $\ket{\psi}$, the exact kernel test using the exact block encoding of $K$ gives, as in~\cite[Equation~(22)]{Rudolph26},
\begin{equation}
    \label{eq:newVerifierPacc}
    \Pr{V'_x\text{ accepts }\ket{\psi}}=1-\norm*{K\ket{\psi}}^2.
\end{equation}

To see perfect completeness, the received label $z$ maximizes $\braket{z}{(I+M)^d}{z}$ and thus
\[ m=2^{hd}\braket{z}{(I+M)^d}{z}=\braket{z}{(2^hI+2^hM)^d}{z}\in\bbZ
 \quad\text{and}\quad m\leq2^{(h+1)d}<2^\ell, \]
because $2^hM$ is Hermitian with Gaussian integer entries and $(I+M)^d \preceq 2^dI$. 
A direct calculation shows that the labels $(z,m)$ are valid and realize $\beta=\frac{1}{\braket{z}{(I+M)^d}{z}}$, and the received state $\ket{\psi}=\ket{\psi_z}$ satisfies $K\ket{\psi}=0$. Therefore, $V'_x$ accepts with certainty for \emph{yes} instances.

\vspace{1em}
For \emph{no} instances, invalid labels are rejected, and by convexity it suffices to consider each quantum state $\ket{\psi}$ conditioned on the corresponding measured labels. By \Cref{eq:newVerifierPacc}, it remains to lower-bound $\norm*{K\ket{\psi}}$ for valid labels. 
A direct calculation gives
\[ \norm*{K\ket{\psi}}
    \geq \alpha\norm*{(I-M)^d\ket{\psi}}
       -\alpha\beta\norm*{(I-M^2)^d\ketbra{z}{z}\ket{\psi}} 
    \geq \alpha\rbra[\Big]{1-\frac{1}{2d}}^d-\alpha\beta
    \geq\alpha\rbra[\Big]{\frac{1}{2}-\beta}
    \geq\frac{1}{16}.
\]
Here, the last inequality uses $\alpha \geq 1/4$ and $\beta \leq 1/4$. We thus conclude that $V'_x$ accepts with probability at most $1-1/256$.

\subsection{Discussion and open problems}
\label{subsec:discussion-open-problem}

The main conceptual difference between achieving perfect completeness for \QMA{} and \QCMA{} is that, for \QMA{}, the verifier needs from the prover a dyadic encoding of the \emph{exact matrix element} $\bra{z}(I+M)^d\ket{z}$ appearing in \Cref{eq:intro-qma-matrix-oblique}. For \QCMA{}, completeness gives $p_z\geq c$ for some $z$, so $d=1$ suffices and the matrix element becomes $\bra{z}\rbra*{I+\Lambda(M)}\ket{z}=1+p_z$, where $\Lambda(M)$ is as in \Cref{eq:pacc-QCMA} below. This is an affine function of the \emph{acceptance probability} $p_z$ and aligns with the intuition behind~\cite{NWZ09,JKNN12}. To further illustrate this difference, we provide an alternative proof that \QCMA{} achieves perfect completeness, inspired by \Cref{subsec:qma-proof-overview}.

\paragraph{An alternative proof of $\QCMA=\QCMA_1$.}
We now sketch our construction and prove that $\QCMA[c,s]\subseteq\QCMA\sbra*{1,1-(1-s)^2(c-s)^2/16}$. 
Since we only need to consider classical witnesses, it suffices to consider an effective \emph{diagonal} acceptance operator 
\begin{equation}
    \label{eq:pacc-QCMA}
    \Lambda(M)=\sum_{z\in\binset^w} \ketbra{z}{z}M\ketbra{z}{z} \eqqcolon \sum_{z\in\binset^w} p_z\ketbra{z}{z},
\end{equation} 
where $\Lambda(\cdot)$ denotes the completely dephasing channel. 
This effective operator is realized by coherently copying the witness in the computational basis before running the original verifier. For the remainder of this paragraph, we write $M_\ttC\coloneqq \Lambda(M)$.

Hence, a direct calculation gives $\ket{\psi_z} = \ket{z} = \ket{\phi_z}$, which makes $I-\ketbra{\phi_z}{z} = I-\ketbra{z}{z}$ an \emph{orthogonal} projection with $\norm{I-\ketbra{z}{z}}=1$ and $\ker(I-\ketbra{z}{z})=\spanset\cbra{\ket{z}}$. However, this kernel identity holds for both \emph{yes} and \emph{no} instances, so we must also constrain the scalar supplied by the prover. We choose $d=1$ and consider a matrix $K$ that admits an efficient exact block encoding:
\begin{equation}
    \label{eq:intro-qma-matrix-QCMA}
    K \coloneqq \alpha \rbra*{ I-M_\ttC } - \frac{1}{4}\rbra*{I-M_\ttC^2}\ketbra{z}{z}, \quad\text{where } \alpha \in \sbra*{\frac{1+c}{4},\frac{1}{2}}.
\end{equation}
Here, $\alpha$ is a prover-supplied dyadic number. 
In the resulting $\QCMA_1$ proof system, the classical witness consists of $z$ and an exact encoding of $\alpha$ using $h+2$ fractional bits. The verifier first checks whether the condition $(1+c)/4 \leq \alpha \leq 1/2$ holds, rejecting otherwise, and then prepares $\ket{z}$ and performs the exact kernel test (\Cref{lemma:exact-kernel-test}) with the matrix $K$. 

We then complete the analysis. For \emph{yes} instances, an honest prover sends $z$ with $p_z \geq c$ and sets $\alpha=(1+p_z)/4$, which satisfies $(1+c)/4 \leq \alpha \leq 1/2$ and gives $K\ket{z}=0$. 
For \emph{no} instances, fix any valid pair $(z,\alpha)$. We show that $\sigma_{\min}(K) \geq (1-s)(c-s)/4$ by observing that 
\[ K\ket{z'} = \begin{cases}
    \alpha(1-p_{z'}) \ket{z'}, &\text{if } z'\neq z\\
    \rbra*{ (1-p_z)\rbra[\big]{\alpha - \frac{1+p_z}{4}} } \ket{z}, &\text{otherwise}
\end{cases}. \]
Since $p_{z'} \leq s$ and $K$ is diagonal, we thus obtain $\sigma_{\min}(K) \geq (1-s)(c-s)/4$ as desired.\footnote{This is because $\sigma_{\min}(K) \geq \min\limits_{z'\neq z}\cbra[\big]{ \abs*{\alpha(1-p_{z'})}, (1-p_z)\rbra[\big]{\alpha - \frac{1+p_z}{4}} } \geq \min\cbra[\big]{ \frac{(1+c)(1-s)}{4}, \frac{(1-s)(c-s)}{4}} = \frac{(1-s)(c-s)}{4}$. }

\paragraph{Open problems.} The most intriguing open problem is perhaps whether $\QMAtwo=\QMAtwo_1$:
\begin{enumerate}[label={\upshape(\alph*)}]
    \item\label{question:QMAtwo} Does \QMAtwo{} achieve perfect completeness? 
\end{enumerate}

While the exact kernel test already applies to product witnesses~\cite[Theorem 1.9]{Rudolph26}, an obstacle is that our techniques based on polynomial transformations exploit the \emph{eigenvalues} of the acceptance operator $M$ rather than the \emph{separable values}, namely the maximum expectation of $M$ over product states. This limitation prevents a direct extension of our perfect-completeness argument to product witnesses, suggesting that resolving Question \ref{question:QMAtwo} requires new ideas. 

\vspace{1em}
Moreover, our transformations for \QIPtwo{} and \qqQAM{} apply to two-message \emph{space-bounded} quantum interactive proof systems introduced in~\cite{LGLNW25}, such as $\QIPLtwo$ and its variants, regardless of whether the verifier performs intermediate measurements. Nevertheless, our techniques do not immediately extend to \QMAL{}~\cite{FR21}, which motivates the following question:
\begin{enumerate}[label={\upshape(\alph*)}]
    \setcounter{enumi}{1}
    \item\label{question:QMAL} Does \QMAL{} achieve perfect completeness? 
\end{enumerate}

While both the original verifier and the witness use $O(\log{n})$ qubits, the classical label length $\ell=(h+1)d+1$ can still be \emph{polynomial} in $n$, since the verifier may contain \emph{polynomially} many Hadamard gates. Therefore, resolving Question \ref{question:QMAL} would clarify whether our clock-free construction can preserve the original verifier's $O(\log n)$-qubit space. 

\subsection{Related work}

\paragraph{Perfect completeness in the precise setting. }
The investigation of quantum interactive proof systems in the precise setting began in~\cite{IKW12}, which established the equivalence between $\QIP(3)_1$ with \emph{doubly-exponentially small} promise gap and \EXP{}. For one-message (quantum) proof systems, considering \emph{exponentially small} promise gap strengthens $\QMA_1$ to \PSPACE{}:
\[ \PreciseQMA_1 = \PreciseQMA = \PSPACE = \mathsf{PreciseStoqMA}_1. \]
Here, the first two equivalences were proven in~\cite{FL16,FL18}, with a simpler proof given in~\cite{Li22}, while the last equivalence was established in~\cite{LW26}, where \StoqMA{} is a subclass of \QMA{} that lies between \MA{} and \AM{}~\cite{BBT06}. 
Similarly, when the witness is restricted to \emph{product states}, considering exponentially small promise gap strengthens $\QMAtwo_1$ to \NEXP{}:
\[ \PreciseQMA(2)_1 = \PreciseQMA(2) = \NEXP = \mathsf{PreciseStoqMA}(2). \]
Here, the first two equivalences were proven in~\cite{BT07,Pereszlenyi12}, while the last equivalence was established in~\cite{LW26}, with $\StoqMA(2)$ similarly restricting the witness in \StoqMA{} to a product state. Interestingly, $\StoqMA(2)$ collapses to \StoqMA{}~\cite{GJ26} and the underlying multi-to-single witness reduction preserves completeness, leading to $\StoqMA(2)_1=\StoqMA_1$. By comparison, perfect completeness cannot be achieved for $\mathsf{PreciseStoqMA}(2)$ unless $\EXP = \NEXP$: indeed, $\mathsf{PreciseStoqMA}(2)_1 \subseteq \EXP$, as proven in~\cite{LW26}. 

\paragraph{Natural $\QMA_1$-complete problems. }
The most famous $\QMA_1$-complete problem is perhaps \textsc{Quantum $k$-SAT} with $k=4$ in~\cite{Bravyi11}, later strengthened to \textsc{Quantum $3$-SAT}~\cite{GN16}, both of which ask whether a local Hamiltonian on \emph{qubits} is frustration-free. While \textsc{Quantum $2$-SAT} on qubits admits an efficient classical algorithm~\cite{Bravyi11}, \textsc{Quantum $2$-SAT} on low-dimensional systems can remain $\QMA_1$-complete~\cite{ER08,RGN25}.
Further examples include \textsc{XZ-Quantum-$6$-SAT}~\cite{MN26}, where each local projector is diagonal in either the computational basis or the Hadamard basis, and a variant of the $2$-local Hamiltonian problem that tests whether a nontrivial kernel exists~\cite{Rudolph26}.\footnote{This promise problem decides whether a $2$-local Hamiltonian has a nontrivial kernel, with the promise that its smallest singular value is at least some inverse polynomial for \emph{no} instances.}
Topological problems provide another family of examples: \textsc{Clique Homology} without a spectral gap promise is $\QMAone$-hard~\cite{CK24}, while its weighted gapped version was studied in~\cite{KK26}, with a gap in the hardness argument addressed by a subsequent complete proof~\cite{Barak26}.
The $\QMAone$ containment for the weighted gapped problem was established in~\cite{Rudolph26}. More recently, the unweighted gapped version was also shown to be $\QMAone$-complete~\cite{Hayakawa26}. 

\paragraph{Concurrent work.} While preparing this manuscript, we became aware of concurrent work that establishes that \QMA{} achieves perfect completeness~\cite{GR26}, which was released on arXiv one day before our work.
The $\QMA_1$ proof systems in~\cite{GR26} and in our first version~\cite{LV26} are built from \emph{the same unnormalized history vector}. The mechanisms for achieving perfect completeness are nevertheless \emph{different}: (i) in the construction of~\cite{GR26}, the endpoint of the history vector is connected back to its starting point; (ii) the construction of~\cite{LV26}, as recapped in \Cref{sec:clock-construction}, uses the total squared norm of the history vector to determine the rank-one term in a perturbation of $BB^\dagger$, so that the resulting matrix has an exact kernel. The resulting \emph{honest witnesses} and soundness analyses are also \emph{different}.

In addition, our new $\QMA_1$ proof system in \Cref{sec:qma} further simplifies the construction in our first version~\cite{LV26} by removing the \emph{clock register} from the witness state, a register that also appears in~\cite{GR26}. The construction also serves a purpose similar to that of the shifted-OR argument used to achieve perfect completeness for \MA{}~\cite{Lautemann83,ZF87}.


\section{Preliminaries}
\label{sec:preliminaries}

We assume that the reader has basic familiarity with quantum information and computation, and the textbook~\cite{NC10} serves as a good starting point. For a more comprehensive survey of quantum interactive proof systems, we refer the reader to~\cite{VW16}. 

Throughout the paper, all logarithms are base-$2$, and we write $\bbZ[i] \coloneqq \cbra{a+bi \colon a,b\in\bbZ}$ for the Gaussian integers. 
The notation $\norm{\ket{\psi}}$ denotes the Euclidean norm, while $\norm{A}$ denotes the operator norm of a matrix $A$, so that $\norm{A}=\sigma_{\max}(A)$, where $\sigma_{\max}(A)$ denotes the largest singular value of $A$. We set $\abs{A}=(A^\dagger A)^{1/2}$ and write $A\succeq B$ when $A-B$ is positive semidefinite. If $A\succeq 0$, its operator norm coincides with its largest eigenvalue. We also adopt the following conventions: quantum registers are denoted by capital sans-serif letters, such as $\sfM$, and $\ket{\bar{0}}$ denotes $\ket{0}^{\otimes a}$ for $a>1$. 

\subsection{Proof system and gateset conventions}
\label{subsec:proof-system-conventions}
We call $\calI = (\calI_{\yes}, \calI_{\no})$ a \textit{promise problem}, if it satisfies that $\calI_{\yes} \cap \calI_{\no} =\emptyset$ and $\calI_{\yes} \cup \calI_{\no} \subseteq \binset^*$. For simplicity, we write $x\in\calI$ for $x\in \calI_{\yes} \cup \calI_{\no}$. 

Given a promise problem $\calI$, we write $(\protocol{P}{V})(x)$ for an $m(|x|)$-message quantum proof system on input $x\in \calI$ , where $m=2\ell$ is \emph{even}, the verifier's actions are $V(x)_1,\cdots,V(x)_{\ell+1}$, and the prover's actions are $P(x)_1,\cdots,P(x)_\ell$. The verifier and prover alternately apply their actions, with the last message received by the verifier. For convenience, we write $V_j \coloneqq V(x)_j$ and $P_j \coloneqq P(x)_j$. When a transformation maps a proof system $\protocol{P}{V}$ to a new proof system $\protocol{P'}{V'}$, we may omit the input $x$ and the underlying promise problem $\calI$.

\paragraph{Gateset conventions.}
For the one-message setting (\QMA{}) and the classical-public-coin two-message setting (\QAM{}), one-sided errors are achieved over the \emph{fixed exact gateset}:\footnote{Equivalent partial characterizations that slightly relax this restriction are given in~\cite{Rudolph26} by defining $\QMA_1$ over fixed \emph{number fields} $\bbF$, allowing a number-field-dependent universal gateset and other gates implementable over $\bbF$.}
\[ \calG\coloneqq\cbra{\Had,X,\Toffoli,S}, \quad\text{where } S\coloneqq\diag(1,i).\] 

Given a verification circuit $V_x$, where $x\in\calI$ for $\calI \in \QMA^\calG$, the corresponding acceptance operator $M \coloneqq \bra{\bar{0}} V_x^\dagger \ketbra{1}{1}_\Out V_x \ket{\bar{0}}$ is a $2^w\times 2^w$ positive semidefinite matrix, where $w$ denotes the number of qubits in the message (i.e., witness) register $\sfM$. 
This gateset-dependent definition of $\QMA^\calG$ has the following properties:
\begin{enumerate}[label={\upshape(\arabic*)}]
    \item Parallel error reduction (e.g., ~\cite[Theorem~3.2]{MW05}) is exactly implementable for $\QMA^\calG$, since the required Boolean operations are exact over $\calG$: a \CNOT{} can be implemented by a \Toffoli{} with a temporary $\ket{1}$ control. 
    \item The inverse of the verification circuit, $V_x^\dagger$, is exactly implementable over $\calG$: all gates other than $S$ are self-inverse, and $S^\dagger=S^3$.    
    \item Each entry of $M$ belongs to $2^{-h} \bbZ[i]$, and thus has polynomial-length bit precision, where $h(n)$ denotes the number of Hadamard gates in $V_x$. This is because each matrix entry of $V_x$ belongs to $2^{-h/2}\bbZ[i]$, while the other gates have Gaussian-integer entries.
\end{enumerate}

\subsection{Elementary transformations for quantum interactive proof systems}
\label{subsec:qip-preliminaries}

We begin with useful information-theoretic lemmas and then present two elementary transformations from the literature. We write $\ket{\Phi^+}= (\ket{00}+\ket{11})/\sqrt{2}$ for an EPR pair. A key information-theoretic quantity is the \emph{(Uhlmann) fidelity} $\F(\rho_0,\rho_1) \coloneqq \Tr\abs{\sqrt{\rho_0}\sqrt{\rho_1}}$, and we need Uhlmann's theorem and the following additional properties:

\begin{lemma}[Uhlmann's theorem~\cite{Uhlmann76}]
    \label{lemma:uhlmann}
    Let $\rho_0$ and $\rho_1$ be quantum states on register $\sfA$.
    Fix a purification $\ket{\psi_0}_{\sfA\sfR}$ of $\rho_0$ and a purification $\ket{\psi'_1}_{\sfA\sfR}$ of $\rho_1$. Then, the squared fidelity satisfies
    \[ \F^2(\rho_0,\rho_1) 
    = \max_{\ket{\psi_1}}\abs*{\innerprod{\psi_0}{\psi_1}}^2 =\max_{U_\sfR}\abs*{\bra{\psi_0} \rbra*{I_\sfA\otimes U_\sfR}\ket{\psi'_1}}^2. \]
    Here, the first maximum ranges over all purifications $\ket{\psi_1}$ of $\rho_1$ on $(\sfA,\sfR)$, while the second ranges over all unitaries $U_\sfR$ on $\sfR$. 
    The second equality implies the \emph{unitary equivalence of purifications}: 
    \[ \Tr_{\sfR}(\ketbra{\psi_1}{\psi_1}) = \rho_1 = \Tr_{\sfR}(\ketbra{\psi'_1}{\psi'_1}), \quad\text{where } \ket{\psi_1} = (I_\sfA\otimes U_\sfR) \ket{\psi'_1} \text{ for some unitary } U_\sfR. \] 
\end{lemma}

\begin{lemma}[Data-processing inequality for the fidelity, adapted from~{\cite[Theorem 9.6]{NC10}}]
    \label{lemma:fidelity-data-processing}
    Let $\rho_0$ and $\rho_1$ be quantum states of the same dimension. For any quantum channel $\calN$, 
    \[ \F\rbra*{\calN(\rho_0), \calN(\rho_1)} \geq \F(\rho_0,\rho_1). \]
\end{lemma}

\begin{lemma}[{\cite[Lemma 2]{SR01} \&~\cite[Lemma 3.3]{NS03}}]
    \label{lemma:sum-of-squared-fidelity}
    Let $\rho_0$ and $\rho_1$ be quantum states of the same dimension. Then, for any quantum state $\xi$ of the same dimension, 
    \[\F^2(\rho_0,\xi) + \F^2(\xi,\rho_1) \leq 1 + \F(\rho_0,\rho_1).\]
\end{lemma}

The following corollary of \Cref{lemma:uhlmann,lemma:fidelity-data-processing} is also useful:

\begin{corollary}[Fidelity extension identity, adapted from~{\cite[Theorem~3.28]{Watrous18}}]
    \label{corr:fidelity-extension-identity}
    Let $\rho_0$ be a quantum state on registers $(\sfA,\sfB)$ and let $\tau$ be a quantum state on register $\sfA$. Then,
    \[ \F^2\rbra*{\Tr_{\sfB}(\rho_0),\tau} = \max_{\rho_1:\,\Tr_{\sfB}(\rho_1)=\tau}\F^2(\rho_0,\rho_1). \]
    Here, the maximum ranges over quantum states $\rho_1$ on registers $(\sfA,\sfB)$. 
\end{corollary}

\subsubsection{Single-qubit exact-half calibration from perfect rewindability}
\label{subsec:calibration}

The notion of \emph{perfectly rewindable} quantum interactive proof systems was originally introduced in~\cite[Section 3]{KKMV09} for the multi-prover setting, and later specialized to the single-prover setting in~\cite[Section 7.2]{KLGN15}. This transformation \emph{calibrates} the acceptance probability to $1/2$ for a specific prover for \emph{yes} instances in the resulting proof system, while the prover adds only \emph{one new qubit} to the last message: 

\begin{lemma}[Single-qubit exact-half calibration]
\label{lemma:exact-half-calibration}
    Let $\protocol{P}{V}$ be an $m$-message quantum proof system with completeness at least $2/3$ and soundness at most $1/3$, where $m$ is even. There exists an explicit $m$-message quantum proof system $\protocol{P'}{V'}$ with completeness exactly $1/2$ and soundness at most $1/3$. 
\end{lemma}

\begin{proof}
    The only difference between the resulting proof system $\protocol{P'}{V'}$ and the original proof system $\protocol{P}{V}$ is that $P'$ additionally \emph{sends one qubit} in the register $\sfA$ with the last message, and $V'$ accepts only if both $\sfA$ is measured to be $1$ and $V$ accepts. We then complete the analysis:
    \begin{itemize}
        \item For \emph{yes} instances, there exists a prover $P$ such that $p \coloneq \Pr{\protocol{P}{V} \text{ accepts}} \geq 2/3$. The new prover $P'$ prepares a single-qubit state $\rho_\sfA$ such that $\Tr\rbra*{\rho_\sfA\ketbra{1}{1}_\sfA} = 1/(2p)$, giving 
        \begin{align*}
            \Pr{\protocol{P'}{V'} \text{ accepts}} &= \Pr{\rho_\sfA \text{ is measured to be } 1} \cdot \Pr{\protocol{P}{V} \text{ accepts}}\\
            &= \Tr(\rho_\sfA \ketbra{1}{1}_\sfA) \cdot p = \frac{1}{2p} \cdot p = \frac{1}{2}.
        \end{align*}
        Here, the choice of $p$ is \emph{prover-dependent} and unknown to the verifier $V'$. 
        \item For \emph{no} instances, let $\widehat{\rho}$ be the joint state returned by an arbitrary prover, which may be entangled across $\sfA$ and the remaining qubits in $\sfM$. Let $E$ be the final acceptance effect of $V$ on the message and the verifier's private registers, with $0 \preceq E \preceq I$. Then,
        \begin{align*}
            \Pr{\protocol{P'}{V'} \text{ accepts}} &= \Tr\rbra*{ \rbra*{\ketbra{1}{1}_\sfA \otimes E} \widehat{\rho} }\\
            &\leq \Tr\rbra*{ \rbra*{I_\sfA \otimes E} \widehat{\rho} }\\
            &= \Tr\rbra*{E \ \Tr_\sfA\rbra*{ \widehat{\rho} }} 
            = \Pr{\protocol{P}{V} \text{ accepts}} \leq \frac{1}{3}. 
        \end{align*}
        Here, the second line follows from $\ketbra{1}{1}_{\sfA}\otimes E\preceq I_{\sfA}\otimes E$, while tracing out $\sfA$ in the last line turns the final action of $P'$ into a legal action for the original proof system.  \qedhere
    \end{itemize}
\end{proof}

\subsubsection{EPR-pair completion and an efficiently preparable terminal state}
\label{subsec:epr-completion}

We next specialize the technique for achieving perfect completeness in~\cite[Section 3]{KW00} to the calibrated proof system with completeness exactly $1/2$:

\begin{lemma}[EPR-pair completion]
    \label{lemma:EPR-completion}
    Let $\protocol{P}{V}$ be an $m$-message quantum proof system with completeness exactly $1/2$ and soundness at most $s < 1/2$, where $m$ is even. There exists an explicit $(m+2)$-message quantum proof system $\protocol{P'}{V'}$ with \emph{perfect} completeness and soundness at most  $g_\EPR(s)\coloneqq\frac{1}{2}+\sqrt{s(1-s)}$.
\end{lemma}

\begin{proof}
    In the new proof system $\protocol{P'}{V'}$, $V'$ executes the original proof system $\protocol{P}{V}$ without the final measurement and coherently records the decision in a single-qubit register $\sfZ$. The new proof system $\protocol{P'}{V'}$ then adds two messages:
    \begin{enumerate}[label={\upshape(\arabic*)}]
        \item $V'$ sends all qubits in the message and its private registers that purify the state in $\sfZ$. 
        \item $V'$ receives a single-qubit register $\sfZ'$, and projects the state on $(\sfZ,\sfZ')$ onto $\ket{\Phi^+}$. 
    \end{enumerate}

    We now analyze $\protocol{P'}{V'}$:
    \begin{itemize}
        \item For \emph{yes} instances, since $V$ accepts with probability exactly $1/2$, the coherent copy in $\sfZ$ is $I_\sfZ/2$, a single-qubit maximally mixed state. Since $P'$ holds all qubits purifying $I_\sfZ/2$, unitary equivalence of purifications (\Cref{lemma:uhlmann}) allows $P'$ to complete this marginal to an EPR pair,\footnote{A more detailed analysis can be found in~\cite[Section 4.2.1]{VW16}, particularly Equations (4.17) through (4.20).} so $V'$ accepts with certainty. 
        \item For \emph{no} instances, since $V$ accepts with probability $p \leq s < 1/2$, the coherent copy in $\sfZ$ is $(1-p)\ketbra{0}{0}+p\ketbra{1}{1}$, and this marginal is preserved by the final prover action due to \Cref{lemma:uhlmann}. Let $\sigma_{\sfZ\sfZ'}$ be the quantum state before the final projection. By the data-processing inequality for fidelity (\Cref{lemma:fidelity-data-processing}), a direct calculation gives 
        \begin{subequations}
        \label{eq:epr-value}
        \begin{align}
            \Pr{\protocol{P'}{V'} \text{ accepts}} &= \F^2(\ketbra{\Phi^+}{\Phi^+}, \sigma_{\sfZ\sfZ'})\\
            &\leq \F^2\rbra*{ \begin{pmatrix} 1-p & 0\\ 0 & p \end{pmatrix}, \frac{I_\sfZ}{2} }
            = \frac{1}{2}+\sqrt{p(1-p)} \eqqcolon g_\EPR(p).
        \end{align} 
        \end{subequations}
        Since $g_\EPR(p)$ is monotonically increasing for $p\in[0,1/2]$ and $p \leq s$, $V'$ accepts with probability at most $g_\EPR(s)$. \qedhere
    \end{itemize}
\end{proof}

By inspecting the proof of \Cref{lemma:EPR-completion}, we notice the following key observation: 
\begin{fact}[Efficiently preparable terminal state]
    \label{fact:EPR-terminal-state}
    In the resulting proof system underlying \Cref{lemma:EPR-completion}, the terminal state on $(\sfZ,\sfZ')$ for \emph{yes} instances is exactly the EPR pair $\ket{\Phi^+}$, which is efficiently preparable. 
\end{fact}

\subsection{Algorithmic toolkit for achieving perfect completeness}
\label{subsec:qma-preliminaries}

We begin with several useful lemmas for quantum computation with one-sided errors that appeared in the context of quantum singular value transformation (QSVT), although \emph{no} actual QSVT techniques are involved in our work. The underlying techniques were previously used in the context of $\QMA_1$~\cite{Rudolph26} and quantum logspace with one-sided errors~\cite{LGLW23}.

\begin{definition}[Exact block encoding, adapted from~{\cite[Definition~24]{GSLW19}}]
    \label{def:block-encoding}
    A unitary $U$ is an \emph{exact block encoding} of $A$ if $\rbra*{ \bra{0}^{\otimes a}\otimes I } U \rbra*{ \ket{0}^{\otimes a}\otimes I }=A$, where $a$ denotes the number of ancillary qubits initialized to $\ket{\bar{0}}$. For convenience, any normalization factor is absorbed into the encoded matrix $A$, with $\norm{A} \leq 1$. 
\end{definition}

\begin{lemma}[Product of exactly block-encoded matrices, adapted from~{\cite[Lemma~30]{GSLW19}}] 
    \label{lemma:block-product}
    Let $U_1$ and $U_2$ be exact block encodings of same-dimensional square matrices $A_1$ and $A_2$, respectively, with disjoint ancillary qubit registers $\sfR_1$ and $\sfR_2$. Then, $(I_{\sfR_2}\otimes U_1)(I_{\sfR_1}\otimes U_2)$ is an exact block encoding of $A_1A_2$.
\end{lemma}

\begin{lemma}[Exact block encoding of POVM operators, adapted from~{\cite[Lemma~26]{GSLW19}}]
    \label{lemma:POVM-encoding}
    Let $0\preceq M\preceq I$ be a POVM operator. Let $U$ be a unitary circuit that uses $a$ ancillary qubits initialized to $\ket{\bar{0}}$ and outputs $0$ in a single-qubit register $\sfO$ with probability $\Tr\rbra*{M\rho}$ for every input state $\rho$.
    With a fresh single-qubit register $\sfF$ initialized to $\ket{0}$, $U^\dagger\CNOT_{\sfO\to\sfF}U$ is an exact block encoding of $M$ using $a+1$ ancillary qubits and one call to each of $U$ and $U^\dagger$.
\end{lemma}

For a unitary $W$, write $\rmC(W)$ for its controlled version. The identity
\begin{equation}
    \label{eq:controlled-conjugation}
    \rmC(U^\dagger\CNOT_{\sfO\to\sfF}U)
    =(I\otimes U^\dagger)\rmC(\CNOT_{\sfO\to\sfF})(I\otimes U)
\end{equation}
shows that only the \CNOT{} requires an additional control.
Thus, if $U$ is exact over $\calG$, this block encoding has an exact controlled implementation by replacing the \CNOT{} with a \Toffoli{} gate.

The following lemma specializes~\cite[Lemma~29]{GSLW19}, using a uniform selector and reversible comparisons to implement dyadic coefficients, and extends the implementation underlying~\cite[Equation 14]{Rudolph26} to $t$ intervals: 
\begin{lemma}[Exact dyadic linear combinations of exactly block-encoded matrices, adapted from~{\cite[Lemma~29]{GSLW19}}]
    \label{lemma:uniform-selector}
    Let $N=2^k$ and let $n_1,\ldots,n_t$ be nonnegative integers such that $\sum_{j=1}^t n_j = N$. Assume that each operator $A_j$ on a common quantum register admits an exact block encoding with an exact controlled implementation.

    \noindent Then, $\frac{1}{N} \sum_{j=1}^t n_j A_j$ admits an exact block encoding using one controlled call per branch and $\poly(t,k)$ gates consisting of reversible integer comparisons and Hadamard gates.
    Ancillary qubits used for reversible comparisons are returned to zero by uncomputation. More generally, signed coefficients can be handled by absorbing each sign into the corresponding branch unitary. 
\end{lemma}

Finally, we use the exact kernel test from~\cite{Rudolph26}, which was originally stated for local Hamiltonians but applies to any matrix admitting an exact block encoding:

\begin{lemma}[Exact kernel test, adapted from~{\cite[Lemma~3.8 and Equation~(22)]{Rudolph26}}]
    \label{lemma:exact-kernel-test}
    Let $U$ be an exact block encoding of $A$ using $a$ ancillary qubits initialized to $\ket{\bar{0}}$. The exact kernel test $\calT_A$ applies $U$ to $\ket{0}^{\otimes a}\ket{\psi}$, and rejects if measuring the ancillary qubits in the computational basis yields the all-zero outcome. For every quantum state $\ket{\psi}$, 
    \[ \Pr{\calT_A \text{ rejects } \ket{\psi}} = \norm{ A \ket{\psi} }^2. \]
    In particular, every quantum state (i.e., unit vector) in $\ker(A)$ is accepted with certainty. 
\end{lemma}


\section{An endpoint-inward turn-halving transformation for quantum proof systems}
\label{sec:terminal-turn-halving}

The standard \emph{middle-outward} turn-halving transformation in~\cite[Section 4]{KKMV09} starts from a prover-supplied midpoint state, and the resulting proof systems use at least three messages. In this section, we establish an \emph{endpoint-inward} turn-halving transformation, assuming that both endpoint states are \emph{efficiently preparable} by the verifier. This extra resource allows the verifier to coherently prepare both endpoints and simultaneously evolve the two branches, one forward and one backward, in superposition. Consequently, our transformation halves an even number of messages \emph{without retaining the extra public-coin turn.} 

\subsection{The endpoint-inward turn-halving transformation}
\label{subsec:endpoint-fold}

\begin{theorem}[Endpoint-inward turn-halving transformation]
    \label{thm:inward-turn-halving}
    There is a polynomial-time transformation that, given as input the circuit descriptions of an interactive verifier $V$ and the terminal-state preparation circuit, returns a circuit description of an interactive verifier $V'$ such that the following holds. 
    Suppose that the quantum proof system $\protocol{P}{V}$ has $4r$ messages, for some $r\geq 1$, perfect completeness and soundness at most $s$. Suppose further that there is an honest prover $P$ such that, in the case of a \emph{yes}-instance, the terminal state on the message and verifier's private registers $(\sfM,\sfW)$ before the final measurement is an \emph{efficiently preparable} pure state.\footnote{The sizes of these registers $\sfM$ and $\sfW$ may decrease during the interaction, since the prover's actions may be \emph{partial isometries} rather than unitaries.} 
    On every promised input, the preparation circuit produces a state accepted with certainty by the final measurement.
    Then, the quantum proof system $\protocol{P'}{V'}$ has $2r$ messages, perfect completeness and soundness at most $g_\inward(s) \coloneqq (1+\sqrt s)/2$.
\end{theorem}

\begin{proof}
    We start by presenting the new proof system $\protocol{P'}{V'}$. $V'$ uses two single-qubit registers $\sfB$ and $\sfB'$ to label the forward and backward branches, keeping $\sfB$ in its private memory and sending $\sfB'$ to the prover. We absorb $V_1$ into the initial preparation and $V_{2r+1}$ into the final measurement, applying $V_{2r+1}^\dagger$ to the prepared terminal state.
    Without loss of generality, we assume that the endpoint states $\ket{\psi_\init}_{\sfM\sfW}$ and $\ket{\psi_\final}_{\sfM\sfW}$ have the same dimension,\footnote{This assumption can be achieved simply by padding with sufficiently many $\ket{0}$ qubits.} and define
    \begin{equation}
        \label{eq:endpoint-superposition}
        \ket{\Xi} \coloneqq \frac{\ket{00}_{\sfB\sfB'}\ket{\psi_\init}_{\sfM\sfW}
            +\ket{11}_{\sfB\sfB'}\ket{\psi_\final}_{\sfM\sfW}}{\sqrt{2}}.
    \end{equation}

    Let $\Pi_{\match} \coloneqq \ketbra{00}{00}_{\sfB\sfB'}+\ketbra{11}{11}_{\sfB\sfB'}$, where $\Pi_\match$ projects onto $\spanset\cbra*{\ket{00}_{\sfB\sfB'},\ket{11}_{\sfB\sfB'}}$. The new proof system $\protocol{P'}{V'}$ is described in \Cref{protocol:inward-turn-halving}. 

    \begin{algorithm}[ht!]
    \SetEndCharOfAlgoLine{.}
    \SetAlgoVlined
    \setlength{\parskip}{6pt}
    \SetKwFor{For}{For}{:}{}
    \SetKwComment{Comment}{// }{}

    \smallskip
    \textbf{1.} $V'$ prepares $\ket{\Xi}$ from \Cref{eq:endpoint-superposition} and sends $(\sfB',\sfM)$\;

    \BlankLine
    {\color{gray}\Comment{Execute $\protocol{P}{V}$ simultaneously forward and backward from the endpoints.}}
    \textbf{2.} \For{$j \leftarrow1$ \KwTo $r-1$}{
        \textbf{2.1} $V'$ receives $(\sfB',\sfM)$, possibly modified by $P'$\;
        {\color{gray}\Comment{Keep the two branches isolated using $\Pi_\match$.}}
        \textbf{2.2} $V'$ projects $(\sfB,\sfB')$ onto $\spanset\cbra*{\ket{00},\ket{11}}$ and rejects on the orthogonal outcome\;
        \textbf{2.3} $V'$ applies $\ketbra{0}{0}_{\sfB}\otimes V_{j+1}
         +\ketbra{1}{1}_{\sfB}\otimes V_{2r+1-j}^{\dagger}$ to $(\sfB, \sfM, \sfW)$, and sends $(\sfB',\sfM)$\;
    }
    \textbf{3.} $V'$ receives $(\sfB',\sfM)$, and applies $\ketbra{0}{0}_{\sfB}\otimes V_{r+1}
     +\ketbra{1}{1}_{\sfB}\otimes I_{\sfM\sfW}$ to $(\sfB, \sfM, \sfW)$\;

    \BlankLine
    {\color{gray}\Comment{Accept if the two branches meet at the midpoint.}}
    \textbf{4.} $V'$ projects $(\sfB,\sfB')$ onto $\ket{\Phi^+}$ and accepts on this outcome.
    \BlankLine
    \caption{Inward turn-halving transformation from $4r$ messages to $2r$ messages.}
    \label[algorithm]{protocol:inward-turn-halving}
    \end{algorithm}

    We now analyze $\protocol{P'}{V'}$. 

    \begin{itemize}
        \item For \emph{yes} instances, let partial isometries $P_1,\ldots,P_{2r}$ denote the prover's actions on its private and message registers $(\sfP,\sfM)$ in the original proof system $\protocol{P}{V}$, where $P_j$ denotes the $j$-th action. Let $\ket{\chi_\init}$ be the initial state in the prover's private register. Perfect completeness and the pure terminal-state assumption guarantee that
        \[ P_{2r}V_{2r}P_{2r-1}\cdots V_2P_1
        \rbra[\big]{\ket{\psi_\init}\ket{\chi_\init}}
        =\ket{\psi_\final}\ket{\chi_\final} \]
        for some pure private state $\ket{\chi_\final}$. Let $\ket{\psi_\midd}_{\sfP\sfM\sfW}$ be the global pure state immediately after $V_{r+1}$, with (reduced) midpoint state $\rho_\midd=\Tr_{\sfP}\ketbra{\psi_\midd}{\psi_\midd}$. 

        The honest prover $P'$ coherently prepares $\ket{\chi_\init}$ in the forward branch ($\sfB'=0$) and $\ket{\chi_\final}$ in the backward branch ($\sfB'=1$), and, at each turn $j=1,\ldots,r$, applies $P_j$ in the forward branch and $P_{2r+1-j}^\dagger$ in the backward branch.
        Consequently, both branches reach the state $\ket{\psi_\midd}$. Since $P'$ never changes $\sfB'$, the joint state on $(\sfB,\sfB')$ remains supported on $\spanset\cbra*{\ket{00},\ket{11}}$, so each intermediate projection leaves the joint state unchanged. Immediately before the final projection $\Pi_\EPR\coloneqq \ketbra{\Phi^+}{\Phi^+}$, the state is
        \[ \frac{\ket{00}_{\sfB\sfB'}\ket{\psi_\midd}
              +\ket{11}_{\sfB\sfB'}\ket{\psi_\midd}}{\sqrt2}
            =\ket{\Phi^+}_{\sfB\sfB'}\ket{\psi_\midd}. \]
        Therefore, applying $\Pi_\EPR$ leaves the state unchanged, and $V'$ accepts with certainty. 
        
        \item For \emph{no} instances, let partial isometries $P'_1,\ldots,P'_r$ denote an arbitrary cheating prover's actions on $(\sfP,\sfM,\sfB')$, with its initial private state $\ket{\chi}$. Define the branch-diagonal blocks of $P'_j$ by 
        \[ Q_j^{\tt 0}=(\bra{0}_{\sfB'}\otimes I)P'_j(\ket{0}_{\sfB'}\otimes I)
        \quad\text{and}\quad
        Q_j^{\tt 1}=(\bra{1}_{\sfB'}\otimes I)P'_j(\ket{1}_{\sfB'}\otimes I). \]
        Also define the composite forward and backward operators:
        \[ E^{\tt 0} \coloneqq V_{r+1}Q_r^{\tt 0}V_rQ_{r-1}^{\tt 0}\cdots V_2Q_1^{\tt 0}
        \quad\text{and}\quad
        E^{\tt 1} \coloneqq Q_r^{\tt 1}V_{r+2}^{\dagger}Q_{r-1}^{\tt 1}V_{r+3}^{\dagger} \cdots Q_2^{\tt 1}V_{2r}^{\dagger}Q_1^{\tt 1}, \]
        where $E^{\tt 1}=Q_1^{\tt 1}$ when $r=1$. The corresponding  forward and backward \emph{unnormalized} states are thus defined by 
        \[ \ket{\xi^{\tt 0}} \coloneqq E^{\tt 0}(\ket{\psi_\init}\ket{\chi})
            \quad\text{and}\quad
            \ket{\xi^{\tt 1}} \coloneqq E^{\tt 1}(\ket{\psi_\final}\ket{\chi}). \]
        Since $E^{\tt 0}$ and $E^{\tt 1}$ are diagonal blocks of partial isometries, it follows that $\norm{\ket{\xi^{\tt 0}}}\leq 1$ and $\norm{\ket{\xi^{\tt 1}}}\leq 1$. Consequently, the unnormalized component that survives the projection $\Pi_\match$ and can overlap with $\Pi_\EPR$ is $\rbra*{\ket{00}_{\sfB\sfB'}\ket{\xi^{\tt 0}}+\ket{11}_{\sfB\sfB'}\ket{\xi^{\tt 1}}}/\sqrt{2}$, giving
        \begin{equation}
            \label{eq:bell-cross}
            \Pr{\protocol{P'}{V'} \text{ accepts}} 
            = \frac{1}{4}\norm*{\ket{\xi^{\tt 0}}+\ket{\xi^{\tt 1}}}^{2}
            \leq\frac{1}{2}\rbra*{1+\abs*{\innerprod{\xi^{\tt 1}}{\xi^{\tt 0}}}}.
        \end{equation}

        To establish the soundness bound, it remains to bound the cross term $\abs*{\innerprod{\xi^{\tt 1}}{\xi^{\tt 0}}}$. 
        By purification and padding, take $P'_1,\ldots,P'_r$ to be unitary. In the original proof system, let $P$ start in $\ket{\chi}$ and apply $P'_1,\ldots,P'_r,(P'_r)^\dagger,\ldots,(P'_1)^\dagger$, replacing $\sfB'$ at each turn by a fresh private qubit, initialized to $\ket{0}$ for the first $r$ turns and $\ket{1}$ for the remaining $r$ turns, and leaving each such qubit untouched thereafter.
        Let $\ket{\zeta}$ denote the joint state immediately before the final measurement of the original verifier. Let $\ket{0^r1^r} \coloneqq \ket{0}^{\otimes r}\ket{1}^{\otimes r}$. Then, $\ket{\zeta}$ admits the following decomposition     
        \[ \ket{\zeta} = \rbra*{\rbra*{E^{\tt 1}}^\dagger E^{\tt 0} \rbra*{\ket{\psi_\init}\ket{\chi}}}\otimes \ket{0}^{\otimes r}\ket{1}^{\otimes r} + \ket{\zeta_\perp}. \]
        Here, $\ket{\zeta_\perp}$ denotes the remaining component, whose projection onto $\ket{0^r1^r}$ on the fresh qubits replacing $\sfB'$ at each prover turn is zero.
        
        A direct calculation then gives
        \begin{subequations}
        \label{eq:cross-term}
        \begin{align}
            \abs*{\innerprod{\xi^{\tt 1}}{\xi^{\tt 0}}}^2
            &= \abs*{ \bra{\psi_\final}\bra{\chi} (E^{\tt 1})^\dagger E^{\tt 0} \ket{\psi_\init}\ket{\chi} }^2\\
            &= \abs*{ \rbra*{\bra{\psi_\final}\bra{\chi}\otimes \bra{0^r1^r}} \ket{\zeta} }^2\\
            &\leq \norm*{(\bra{\psi_\final}\otimes I)\ket{\zeta}}^2
            \leq \Pr{\protocol{P}{V}\text{ accepts}}.
        \end{align}
        \end{subequations}

        Combining \Cref{eq:bell-cross,eq:cross-term} with the original soundness bound, we conclude that $V'$ accepts with probability at most $(1+\sqrt{s})/2$. \qedhere
    \end{itemize}
\end{proof}

\subsection{An implication: Achieving perfect completeness for \QIPtwo{}}
\label{subsec:qip2-application}

The main implication of our inward turn-halving transformation is that $\QIPtwo = \QIPtwo_1$, as stated in \Cref{thm:QIP(2)-informal}, enabled by the observation in \Cref{fact:EPR-terminal-state} that the terminal state is exactly an EPR pair. We now prove this implication.

\begin{proof}[Proof of \Cref{thm:QIP(2)-informal}]
    It suffices to show that \QIPtwo{} achieves perfect completeness. Consider a promise problem $\calI=(\calI_\yes,\calI_\no)\in\QIPtwo$ and a corresponding two-message proof system $\protocol{P}{V}$ with completeness at least $2/3$ and soundness at most $1/3$. 

    Applying two elementary transformations to $\protocol{P}{V}$, namely the exact-half calibration (\Cref{lemma:exact-half-calibration}), followed by the EPR-pair completion (\Cref{lemma:EPR-completion}), yields a four-message quantum proof system $\protocol{P'}{V'}$ with perfect completeness and soundness $g_\EPR(1/3)=\frac{1}{2}+\frac{\sqrt{2}}{3}$ such that both the initial state and the terminal state are efficiently preparable. The latter property follows from \Cref{fact:EPR-terminal-state}, obtained by inspecting the proof of \Cref{lemma:EPR-completion}. 
    
    Next, we apply our inward turn-halving transformation in \Cref{thm:inward-turn-halving} (with $r=1$) to obtain a two-message proof system $\protocol{P''}{V''}$ with perfect completeness. For soundness, combining the soundness bounds in \Cref{lemma:EPR-completion,thm:inward-turn-halving} gives
    \[ g_\inward\rbra[\bigg]{g_\EPR\rbra[\bigg]{\frac{1}{3}}} = \frac{1}{2} \rbra[\Bigg]{1+\sqrt{\frac{1}{2}+\frac{\sqrt{2}}{3}}} 
    = \frac{1}{2} \rbra[\bigg]{ 1 + \frac{1+\sqrt{2}}{\sqrt{6}} }
    = s_\star. \]
    Therefore, $V''$ accepts with probability at most $s_\star$, as desired. 
\end{proof}


\section{Achieving perfect completeness for \qqQAM{}}
\label{sec:qq-perfect-completeness}

While the exact-half calibration (\Cref{lemma:exact-half-calibration}) and the EPR-pair completion (\Cref{lemma:EPR-completion}) also apply to \qqQAM{} proof systems, combining them with our inward turn-halving transformation (\Cref{thm:inward-turn-halving}) does not preserve the quantum-public-coin form of the verifier's message: the forward branch sends halves of the original $w$ EPR pairs, $I_\sfW/\dim(\sfW)$, where $\sfW$ is the verifier's private register with the same dimension as the message register $\sfM$, while the backward branch sends a half of the terminal EPR pair, $I_{\sfZ}/2$. Therefore, the approach of \Cref{thm:inward-turn-halving} does not apply directly because the two maximally mixed states can have \emph{different} dimensions.

In this section, we overcome this obstruction by encoding the two endpoint states into a \textsc{Fixed-Target Close Image} instance, which is essentially the \qqQAM{}-complete problem \textsc{Close Image to Totally Mixed} from~\cite{KLGN19}, but formulated in terms of squared fidelity:

\begin{definition}[Fixed-Target Close Image, \FixedCI{}, adapted from~{\cite{KLGN19}}]
    \label{def:FixedCI}
    Let $Q$ be a unitary quantum circuit such that $Q \colon (\sfX,\sfR_\In) \mapsto (\sfY,\sfR_\Out)$, with designated input register $\sfX$ and output register $\sfY$. The registers $\sfR_\In$ and $\sfR_\Out$ consist of all qubits other than those in the input and output registers, respectively. Define the induced channel $\calN$ and the fixed target $\sigma$ by 
    \begin{equation}
        \label{eq:FixedCI-def}
        \calN(\rho) \coloneqq \Tr_{\sfR_\Out}\rbra*{ Q (\rho \otimes \ketbra{\bar{0}}{\bar{0}}_{\sfR_\In}) Q^\dagger } 
        \quad\text{and}\quad
        \sigma \coloneqq \frac{I_\sfY}{\dim(\sfY)}.
    \end{equation}
    Let $a(n)$ and $b(n)$ be efficiently computable functions such that $0 \leq b(n) < a(n) \leq 1$. The promise problem $\FixedCI[a(n),b(n)]$ asks to distinguish between the following two cases: 
    \begin{itemize}
        \item \emph{Yes:} There exists a quantum state $\rho$ on $\sfX$ such that $\F^2\rbra[\big]{\calN(\rho),\sigma}\geq a(n)$.
        \item \emph{No:} Every quantum state $\rho$ on $\sfX$ satisfies $\F^2\rbra[\big]{\calN(\rho),\sigma}\leq b(n)$.
    \end{itemize}
\end{definition}

\subsection{A fixed-target \textsc{Close Image} construction}
\label{subsec:qq-close-image}

To establish \Cref{thm:qqQAM-informal}, we need the following \qqQAM{} proof system of \FixedCI{}: 

\begin{lemma}[$\FixedCI\in\qqQAM$, adapted from~{\cite[Lemma~4.1]{KLGN19}}]
    \label{lemma:FixedCI-in-qqQAM}
    For every $\FixedCI[a,b]$ instance $Q$, there is an explicit \qqQAM{} proof system with completeness $a$ and soundness $b$. 
    
    \noindent In particular, the maximum acceptance probability of the underlying \qqQAM{} proof system $\protocol{P}{V}$, as described in \Cref{protocol:qq-close-image} (cf.~\cite[Figure~2]{KLGN19}), is exactly $\max_{\rho}\F^2\rbra[\big]{\calN(\rho),\sigma}$, where $\calN(\rho)$ and $\sigma$ are specified in \Cref{eq:FixedCI-def}.
\end{lemma}

\begin{algorithm}[ht!]
    \SetEndCharOfAlgoLine{.}
    \SetAlgoVlined
    \setlength{\parskip}{6pt}
    
    \BlankLine
    \textbf{1.} $V$ prepares $y$ EPR pairs on $(\sfY,\sfR)$ and sends the $y$ halves in $\sfR$, where $y=\log\dim(\sfY)$\;

    \BlankLine
    \textbf{2.}  $V$ receives $\sfR_\Out$, possibly modified by $P$\;

    \BlankLine
    \textbf{3.} $V$ applies $Q^\dagger$ to $(\sfY,\sfR_\Out)$, measures $\sfR_\In$ in the computational basis, and accepts when the measurement outcome is all-zero.
    \BlankLine
    \caption{A $\qqQAM[a,b]$ proof system for $\FixedCI[a,b]$.}
    \label[algorithm]{protocol:qq-close-image}
\end{algorithm}

Next, we proceed with the actual proof: 

\begin{proof}[Proof of \Cref{thm:qqQAM-informal}]
Fix a promise problem $\calI\in\qqQAM$ with a corresponding proof system $\protocol{P}{V}$ of completeness at least $2/3$ and soundness at most $1/3$. By \Cref{lemma:FixedCI-in-qqQAM}, for each input $x\in\calI$, it suffices to efficiently construct a $\FixedCI[1,s_\star]$ instance $Q_\qq$, with its induced channel denoted by $\calN_\qq$ and fixed target $\sigma_\qq$.  

\paragraph{Constructing the quantum channel $\calN_\qq$.} 
We begin with a single-output-qubit channel $\calR(\cdot)$ that embeds the single-qubit exact-half calibration (\Cref{lemma:exact-half-calibration}), where the input register is to be specified and the output register is denoted by $\sfZ$. Let $\sfW$ be $V$'s private register containing halves of $w$ EPR pairs, with $\dim(\sfW)=2^w$, and $\sfM'$ be the message register received from $P$. Let $E$ be the acceptance operator corresponding to the verifier's final action on $(\sfM',\sfW)$. 

The designated input register of $\calR$ (and also $\calN_\qq$) is $\sfX\coloneqq(\sfA,\sfM',\sfW)$, where $\sfA$ denotes the single-qubit register added by the calibration (\Cref{lemma:exact-half-calibration}). We now write the calibration channel $\calR(\cdot)$ explicitly: 
\[ \calR(\rho) \coloneqq (1-\gamma_\rho)\ketbra{0}{0}_{\sfZ}+\gamma_\rho\ketbra{1}{1}_{\sfZ} 
\quad\text{where}\quad 
\gamma_\rho\coloneqq\Tr\rbra*{\rbra*{\ketbra{1}{1}_{\sfA}\otimes E}\rho}. \]

Next, we incorporate the inward turn-halving idea into this construction (\Cref{thm:inward-turn-halving}). To this end, we introduce a single-qubit flag register $\sfF$ and set $\sfY=(\sfF,\sfZ,\sfW)$ as the designated output register of $\calN_\qq$, which defines the fixed target state
\[ \sigma_\qq \coloneqq\frac{I_{\sfF}}2\otimes\frac{I_{\sfZ}}2\otimes \frac{I_{\sfW}}{\dim(\sfW)} =\frac{I_{\sfY}}{2^{w+2}}. \]

Similar to the inward turn-halving transformation, we consider both the forward and backward branches in $\calN_\qq$:
\[ \calN_\qq(\rho) \coloneqq
    \frac{1}{2}\ketbra{0}{0}_{\sfF}\otimes\frac{I_{\sfZ}}2\otimes
      \Tr_{\setminus\sfW}(\rho)
     +\frac{1}{2}\ketbra{1}{1}_{\sfF}\otimes\calR(\rho)_{\sfZ}\otimes
      \frac{I_{\sfW}}{\dim(\sfW)} \]
Here, the forward branch ($\sfF=0$) traces out $(\sfA,\sfM')$ in $\rho$, while the backward branch ($\sfF=1$) maps $\rho$ to the calibrated output qubit. 

The calibrated decision circuit for $\calR$, together with EPR-pair preparations for the flag and maximally mixed pads, implements $\calN_\qq$ exactly with a polynomial number of gates. We implement these operations coherently, collecting all zero-initialized non-input qubits in $\sfR_\In$ and all non-output qubits in $\sfR_\Out$.
Denote the resulting unitary quantum circuit by $Q_\qq\colon(\sfX,\sfR_\In)\mapsto(\sfY,\sfR_\Out)$. After initializing $\sfR_\In$ to $\ket{\bar{0}}$ and tracing out
$\sfR_\Out$, this circuit induces $\calN_\qq$. 

\paragraph{Analysis.}
Applying \Cref{protocol:qq-close-image} to $Q_\qq$ yields the \qqQAM{} proof system $\protocol{P'}{V'}$. By \Cref{lemma:FixedCI-in-qqQAM}, the maximum acceptance probability of $\protocol{P'}{V'}$ is exactly $\max_{\rho}\F^2\rbra[\big]{\calN_\qq(\rho),\sigma_\qq}$, where $\rho$ ranges over all quantum states on $\sfX$. 
Using the orthogonality of the flag sectors and multiplicativity of fidelity, a direct calculation shows that
\begin{equation}
    \label{eq:qq-fidelity-split}
    \F\rbra[\big]{\calN_\qq(\rho),\sigma_\qq}
    =\frac12\F\rbra*{\Tr_{\setminus\sfW}(\rho),\frac{I_{\sfW}}{\dim(\sfW)}}
     +\frac12\F\rbra*{\calR(\rho),\frac{I_{\sfZ}}2}.
\end{equation}

It remains to show that $Q_\qq$ is indeed an instance of $\FixedCI[1,s_\star]$: 
\begin{itemize}
    \item For \emph{yes} instances, let $\tau$ be the
    state on $(\sfM',\sfW)$ produced by a prover whose acceptance probability
    in the original proof system is $p\coloneqq\Tr(E\tau)\geq2/3$.
    Since $\sfW$ contains the retained halves of the original EPR pairs,
    $\Tr_{\setminus\sfW}(\tau)=I_{\sfW}/\dim(\sfW)$.
    As in the single-qubit exact-half calibration, choose a state
    $\rho_{\sfA}$ with $\bra{1}\rho_{\sfA}\ket{1}=1/(2p)$ and set
    $\rho=\rho_{\sfA}\otimes\tau$.
    In the forward branch, the state in $\sfW$ is therefore maximally mixed;
    in the backward branch, $\gamma_\rho=p/(2p)=1/2$, so the calibrated output
    $\calR(\rho)$ is $I_{\sfZ}/2$.
    Both branches consequently agree with the corresponding flag sectors
    of $\sigma_\qq$, giving $\calN_\qq(\rho)=\sigma_\qq$ and hence
    \[
        \max_{\rho}\F^2\rbra[\big]{\calN_\qq(\rho),\sigma_\qq}=1.
    \]

    \item For \emph{no} instances, we first consider an arbitrary fixed input state $\rho$ and a state $\sigma$, both on $\sfX$, where $\sigma$ ranges over all states satisfying $\Tr_{\setminus\sfW}(\sigma) = I_{\sfW}/\dim(\sfW)$. 
    By unitary equivalence of purifications (\Cref{lemma:uhlmann}), $P$ can produce $\Tr_{\sfA}(\sigma)$ on $(\sfM',\sfW)$ by acting on its received
    EPR halves and private register. The soundness bound of $\protocol{P}{V}$ therefore gives
    \[ 0 \leq \gamma_\sigma = \Tr\rbra*{\rbra*{\ketbra{1}{1}_{\sfA}\otimes E}\sigma}
         \leq \Tr\rbra*{\rbra*{I_{\sfA}\otimes E}\sigma}
         =\Tr\rbra*{E\Tr_{\sfA}(\sigma)}
         \leq 1/3. \]
    Using the explicit expression for $\calR(\sigma)$, we obtain
    \[
        \max_\sigma\F\rbra*{\calR(\sigma),\frac{I_{\sfZ}}2}
        =\max_\sigma\frac{\sqrt{1-\gamma_\sigma}+\sqrt{\gamma_\sigma}}{\sqrt2}
        \leq \frac{1+\sqrt{2}}{\sqrt{6}}.
    \]
    Here, the last inequality follows because $\sqrt{1-\gamma}+\sqrt{\gamma}$ is increasing on $[0,1/3]$.
    
    Consequently, for the arbitrary input state $\rho$, the identity in \Cref{eq:qq-fidelity-split} yields
    \begin{subequations}
    \begin{align*}
        \F^2\rbra*{\calN_\qq(\rho),\sigma_\qq}
        &=\frac{1}{4}\rbra*{\F\rbra*{\Tr_{\setminus\sfW}(\rho),\frac{I_{\sfW}}{\dim(\sfW)}} 
            +\F\rbra*{\calR(\rho),\frac{I_{\sfZ}}2}}^{\!2}\\
        &\leq\frac{1}{2}\cdot\F^2\rbra*{\Tr_{\setminus\sfW}(\rho),\frac{I_{\sfW}}{\dim(\sfW)}}
            +\frac{1}{2}\cdot \F^2\rbra*{\calR(\rho),\frac{I_{\sfZ}}2}\\
        &=\frac{1}{2} \max_\sigma\rbra*{\F^2(\rho,\sigma)
                      +\F^2\rbra*{\calR(\rho),\frac{I_{\sfZ}}2}}\\
        &\leq \frac{1}{2} \max_\sigma\rbra*{\F^2\rbra*{\calR(\rho),\calR(\sigma)}
                      +\F^2\rbra*{\calR(\rho),\frac{I_{\sfZ}}2}}\\
        &\leq \frac{1}{2} + \frac{1}{2} \cdot \max_\sigma\F\rbra*{\calR(\sigma),\frac{I_{\sfZ}}2}\\
        &\leq \frac{1}{2} \rbra[\bigg]{1 + \frac{1+\sqrt{2}}{\sqrt{6}}}
        =s_\star.
    \end{align*}
    \end{subequations}
    Here, the second line uses the AM--GM inequality, the third line uses the fidelity extension identity (\Cref{corr:fidelity-extension-identity}), the fourth line applies the data-processing inequality (\Cref{lemma:fidelity-data-processing}) to $\calR$, and the fifth line follows from the inequality in \Cref{lemma:sum-of-squared-fidelity}.
\end{itemize}

Therefore, $Q_\qq$ satisfies the \emph{yes} and \emph{no} conditions of
$\FixedCI[1,s_\star]$ when $x\in\calI_\yes$ and $x\in\calI_\no$,
respectively. By \Cref{lemma:FixedCI-in-qqQAM}, the resulting proof
system $\protocol{P'}{V'}$ has perfect completeness and soundness at most
$s_\star$, which establishes $\qqQAM\subseteq\qqQAMone$ as desired.
\end{proof}


\section{Achieving perfect completeness for \QMA{}}
\label{sec:qma}

In this section, we will prove \Cref{thm:QMA-informal} by first constructing a matrix $K$ whose kernel certifies \emph{yes} instances, then implementing the new $\QMA_1$ verification circuit, and finally establishing perfect completeness and the desired soundness bound.

\subsection{Constructing a matrix whose kernel certifies \emph{yes} instances}
\label{subsec:kernel-construction}

Consider $\calI\in\QMA^\calG$ and an input $x\in\calI$, with $n \coloneqq |x|$. We start with a $\QMA^\calG$ verification circuit $V_x$, together with the induced acceptance operator $M$ of dimension $2^w\times 2^w$ as in \Cref{subsec:proof-system-conventions}, where $w$ denotes the number of qubits used by the witness. 

Set $d \coloneqq w+3$. To establish \Cref{thm:QMA-informal}, we construct a matrix $K$ satisfying the following properties under the promise of vanishing errors:\footnote{The assumed threshold parameters can be achieved by (witness-preserving) error reduction for \QMA{}~\cite{MW05}, which can be implemented exactly over $\calG$: this construction uses the original verification circuit and its inverse, coherent records of computational-basis measurements, and reversible Boolean operations.} 
\begin{itemize}
    \item For \emph{yes} instances, $\norm{M} \geq 1-\frac{1}{2d}$ implies that some valid labels $(z,m)$ give $\ker\rbra{K} \neq \cbra{0}$.
    \item For \emph{no} instances, $M \preceq \frac{I}{2d}$ implies that every valid pair of labels $(z,m)$ gives $\sigma_{\min}(K) \geq \frac{1}{16}$.
\end{itemize}

We now describe the matrix $K$ from the acceptance operator $M$ over the original message register $\sfM$. 
For positive coefficients $\alpha$ and $\beta$ to be specified and $z\in\binset^w$, define
\[  K\coloneqq\alpha(I-M)^d\rbra*{I-\beta(I+M)^d\ketbra{z}{z}}. \]

As explained in \Cref{subsec:qma-proof-overview}, the factor $(I-M)^d$ gives OR-type repetition: its squared norm on a $\lambda$-eigenvector of $M$ is $(1-\lambda)^{2d}$, the corresponding rejection probability. The second factor is a rank-one perturbation of $I$ that, when $\beta\braket{z}{(I+M)^d}{z}=1$, annihilates exactly the vectors proportional to $(I+M)^d\ket{z}$.
Notably, the matrix $K$ need not be Hermitian.

\paragraph{Exact coefficient precision.}
Let $h(n)$ be the number of Hadamard gates in $V_x$. By \Cref{subsec:proof-system-conventions}, $2^hM$ has Gaussian integer entries. The matrix $2^{dh}(I+M)^d$ therefore has Gaussian integer entries and, being Hermitian, has real diagonal entries. Consequently, for every $z\in\binset^w$,
\begin{equation}
\label{eq:qma-certificate-integrality}
    2^{hd}\braket{z}{(I+M)^d}{z}=\braket{z}{2^{hd}(I+M)^d}{z}\in\bbZ.
\end{equation}
Moreover, $0\preceq M\preceq I$ gives $I\preceq(I+M)^d\preceq2^dI$. Thus the integer in \Cref{eq:qma-certificate-integrality} is positive and at most $2^{(h+1)d}$. Both $h$ and $d$ are polynomially bounded in $n$.

\subsection{Implementing the new verification circuit}
\label{subsec:proof-system}

We now present the new verifier $V'_x$, which receives $(z,m,\ket{\psi})$ in one message:
\begin{itemize}
    \item $(z,m)$: a $w$-bit label $z$ and an integer $m$ represented with $\ell\coloneqq(h+1)d+1$ bits;
    \item $\ket{\psi}$: a quantum witness on the original $w$-qubit message register $\sfM$.
\end{itemize}

The verifier first measures the label registers and rejects unless
\[ 2^{hd+2}\leq m<2^\ell. \]
For a valid label, let $k$ be the bit length of $m$ excluding leading zeros, so that $2^{k-1}\leq m<2^k$ and $hd+3\leq k\leq\ell$. 
The verifier computes $\alpha \coloneqq m/2^{k+1}$ and $\alpha\beta = 2^{hd-k-1}$, both of which are \emph{dyadic}, where $\beta \coloneqq 2^{hd}/m$. Every valid label satisfies
\begin{equation}
\label{eq:coefficient-bounds}
 \frac14\leq\alpha<\frac12,\qquad
 0<\beta\leq\frac14,\qquad
 \alpha+\alpha\beta=\frac{m+2^{hd}}{2^{k+1}}\leq2\alpha<1.
\end{equation}
The verifier does not check whether $m$ equals the integer in \Cref{eq:qma-certificate-integrality}.

\paragraph{Constructing $U_K$.}
Since $I-M$ and $I+M$ commute, we have
\begin{equation}
    \label{eq:qma-matrix-expanded}
    K=\alpha(I-M)^d-\alpha\beta(I-M^2)^d\ketbra{z}{z}.
\end{equation}
We first construct exact block encodings of $I-M$ and $I-M^2$ as follows:
\begin{itemize}
    \item\textbf{$U_{I-M}$:} Apply \Cref{lemma:POVM-encoding} to the output-$0$ POVM operator of $V_x$.
    \item\textbf{$U_{I-M^2}$:} Let $U_M$ be the exact block encoding of $M$ obtained by applying \Cref{lemma:POVM-encoding} to the output-$1$ POVM operator of $V_x$. Then apply \Cref{lemma:POVM-encoding} to the POVM operator $I-M^2$, corresponding to the event that measuring the ancillary qubits of $U_M$ yields at least one $1$.
\end{itemize}
The circuit implementations of $U_{I-M}$ and $U_{I-M^2}$ use $O(1)$ calls to $V_x$ and $V_x^\dagger$ and have exact controlled implementations by \Cref{eq:controlled-conjugation}.

Next, we construct exact block encodings of $(I-M)^d$ and $(I-M^2)^d\ketbra{z}{z}$ by taking products:
\begin{itemize}
    \item\textbf{$U_{(I-M)^d}$:} Apply \Cref{lemma:block-product} to $d$ copies of $U_{I-M}$ with disjoint ancillary registers. 
    \item\textbf{$U_{(I-M^2)^d\ketbra{z}{z}}$:} Apply \Cref{lemma:block-product} to $d$ copies of $U_{I-M^2}$ and $U_{\ketbra{z}{z}}$ with disjoint ancillary registers, where $U_{\ketbra{z}{z}}$ is an exact block encoding obtained by flipping a fresh flag qubit initialized to $\ket{0}$ on every computational-basis state $\ket{y}_{\sfM}$ with $y\neq z$.
\end{itemize}

With $U_{(I-M)^d}$ and $U_{(I-M^2)^d\ketbra{z}{z}}$ in hand, we are ready to construct an exact block encoding $U_K$ of $K$, with a circuit implementation over $\calG$. Choose integers $N\coloneqq2^{\ell+1}$, $\iota_\alpha\coloneqq m2^{\ell-k}=N\alpha$, and $\iota_\beta\coloneqq2^{hd+\ell-k}=N\alpha\beta$. 
Apply \Cref{lemma:uniform-selector} to the three branch operators
\[
 (I-M)^d,\qquad -(I-M^2)^d\ketbra{z}{z},\qquad 0,
\]
with respective selector interval lengths $\iota_\alpha$, $\iota_\beta$, and $N-\iota_\alpha-\iota_\beta$, which sum to $N$ by \Cref{eq:coefficient-bounds}.
An $(\ell+1)$-qubit register indexes the $N$ labels in these intervals. Since $\bra{0}X\ket{0}=0$, an $X$ gate serves as an exact block encoding of the zero operator. 

Consequently, the circuit implementation of $U_K$ uses $O(d)$ calls to each of $V_x$ and $V_x^\dagger$ and polynomially many one- or three-qubit gates from $\calG$ and ancillary qubits.

\paragraph{Verification circuit $V'_x$.}
The proof system described in \Cref{protocol:qma-kernel-test} uses the exact kernel test $\calT_K$ (\Cref{lemma:exact-kernel-test}), where the register $\sfA$ contains all ancillary qubits of $U_K$.

\begin{algorithm}[ht!]
    \SetEndCharOfAlgoLine{.}
    \SetAlgoVlined
    \setlength{\parskip}{5pt}

    \textbf{Witness:} The classical labels $(z,m)$ and a quantum state $\ket{\psi}$ on $\sfM$\;

    \textbf{1.} Measure both label registers in the computational basis and reject unless the bounds $2^{hd+2}\leq m<2^\ell$ hold\;

    \textbf{2.} Compute the dyadic coefficients $\alpha=m/2^{k+1}$ and $\alpha\beta=2^{hd-k-1}$ and construct $U_K$ using the dyadic linear combination (\Cref{lemma:uniform-selector}) with $U_{(I-M)^d}$ and $U_{(I-M^2)^d\ketbra{z}{z}}$\;

    \textbf{3.} Run $\calT_K$: initialize $\sfA$ to $\ket{\bar{0}}_\sfA$ and apply $U_K$ to $\ket{\bar{0}}_\sfA\otimes\ket{\psi}_\sfM$. Reject if measuring $\sfA$ in the computational basis yields the all-zero outcome, and accept otherwise\;

    \BlankLine
    \caption{The one-message verification circuit $V'_x$.}
    \label[algorithm]{protocol:qma-kernel-test}
\end{algorithm}

\subsection{Analysis of \texorpdfstring{\Cref{protocol:qma-kernel-test}}{Protocol 3}}
\label{subsec:QMA-analysis}

We first compute the acceptance probability of $V'_x$, and then specify the honest prover's labels and quantum witness. Since $U_K$ is an exact block encoding of $K$, applying \Cref{lemma:exact-kernel-test} to $U_K$ yields, for fixed valid labels and any quantum witness $\ket{\psi}$,
\begin{equation}
\label{eq:acceptance}
 \Pr{V'_x\text{ accepts }\ket{\psi}}=1-\norm*{K\ket{\psi}}^2.
\end{equation}

\parheading{Perfect completeness of \Cref{protocol:qma-kernel-test}.}
For \emph{yes} instances, the honest prover sends $z\in\binset^w$ that maximizes $\braket{z}{(I+M)^d}{z}$, together with $m$ and $\ket{\psi_z}$ satisfying
\[
    m=2^{hd}\braket{z}{(I+M)^d}{z} \quad\text{and}\quad
    \ket{\psi_z}=\frac{(I+M)^d\ket{z}}{\norm*{(I+M)^d\ket{z}}}.
\]
We now verify that the classical labels $(z,m)$ are valid by checking the bound $2^{hd+2}\leq m<2^\ell$.  
By \Cref{eq:qma-certificate-integrality}, $m$ is an integer and satisfies $m\leq2^{(h+1)d}<2^\ell$, proving the upper bound. To see the lower bound $m/2^{hd} \geq 4$, we have
\begin{subequations}
\label{eq:qma-honest-label-bound}
\begin{align}
    \frac{m}{2^{hd}}
    =\braket{z}{(I+M)^d}{z}
    &\geq 2^{-w}\Tr\rbra*{(I+M)^d}\\
    &\geq 2^{-w}(1+\norm{M})^d\\
    &\geq 2^{d-w}\rbra*{1-\frac{1}{4d}}^d\\
    &\geq 8\rbra*{1-\frac14} = 6.
\end{align}
\end{subequations}
Here, the first inequality uses the choice of \(z\) for which the corresponding diagonal entry is the largest and hence at least the average, the second line holds because $\lambda_{\max}(I+M) = 1+\norm{M}$, the third line follows from the completeness condition, and the last line uses the fact that $d=w+3$ and the inequality $(1-t)^d\geq 1-dt$ for $t\in[0,1]$ and integers $d\geq 1$.

We now verify that the state $\ket{\psi}=\ket{\psi_z}$ in the honest witness lies in $\ker(K)$. Applying $K$ to the unnormalized vector $(I+M)^d\ket{z}$ corresponding to $\ket{\psi_z}$, we obtain
\begin{align*}
    K(I+M)^d\ket{z}
    &=\alpha\rbra*{1-\beta\braket{z}{(I+M)^d}{z}} (I-M)^d(I+M)^d\ket{z}\\
    &=\alpha\rbra*{1- \frac{2^{hd}}{m}\cdot\frac{m}{2^{hd}} } (I-M)^d(I+M)^d\ket{z}\\
    &=0.
\end{align*}
Substituting $K\ket{\psi_z}=0$ into \Cref{eq:acceptance} gives $\Pr{V'_x\text{ accepts }\ket{\psi_z}}=1$, as desired.

\parheading{Soundness of \Cref{protocol:qma-kernel-test}.}
For \emph{no} instances, since invalid labels are rejected, it suffices to fix any valid measured labels and, by linearity, consider a pure quantum state $\ket{\psi}$ on $\sfM$ conditioned on these labels.
From \Cref{eq:qma-matrix-expanded} and the soundness condition $M \preceq I/(2d)$, we obtain
\begin{subequations}
\label{eq:qma-soundness-bound}
\begin{align}
 \norm*{K\ket{\psi}}
 &\geq\alpha\norm*{(I-M)^d\ket{\psi}}
       -\alpha\beta\norm*{(I-M^2)^d\ketbra{z}{z}\ket{\psi}}\\
 &\geq\alpha\rbra*{1-\frac{1}{2d}}^d-\alpha\beta\\
 &\geq\alpha\rbra*{\frac12-\beta}\\
 &\geq\frac14\rbra*{\frac12-\frac14}=\frac1{16}.
\end{align}
\end{subequations}
Here, the first line follows from the reverse triangle inequality, the second line uses the fact that $(I-M)^d\succeq(1-1/(2d))^dI$ and $0\preceq I-M^2\preceq I$, the third line follows from $(1-t)^d\geq1-dt$ for $t\in[0,1]$ and integers $d\geq1$, and the last line follows directly from \Cref{eq:coefficient-bounds}.

Since \Cref{eq:qma-soundness-bound} holds for every quantum state $\ket{\psi}$, $\sigma_{\min}(K)\geq1/16$. Using \Cref{eq:acceptance}, we therefore obtain
\begin{equation}
    \label{eq:QMA-soundness-bound}
    \Pr{V'_x\text{ accepts }\ket{\psi}} \leq 1-\rbra*{\frac1{16}}^2=1-\frac1{256}.
\end{equation}
Averaging over the measured labels then proves the same bound as \Cref{eq:QMA-soundness-bound} for arbitrary prover messages, establishing $\QMA^\calG\subseteq\QMA^\calG_1$ as desired.


\section*{Acknowledgments}
\noindent
We thank Baocheng Sun for suggesting and exploring the possibility of removing the clock register in the resulting $\QMA_1$ verification circuit in the preliminary version of our work~\cite{LV26}, which was one of the starting points for the new proof in \Cref{sec:qma}.

The authors were supported in part by funding from the Swiss State Secretariat for Education, Research and Innovation (SERI). 
This work was also supported in part by a grant of access to OpenAI models through the ChatGPT for Academic Researchers program.

\section*{AI use disclosure}
\noindent
Large language model tools were used interactively at all stages of the preparation of this manuscript: from the earliest exploratory research stages to the final exposition and writing steps. The final manuscript was substantially revised by the authors, who remain solely responsible for all mathematical claims, proofs, references, and conclusions.

\bibliographystyle{alphaurlQ}
\bibliography{perfect-completeness}

\appendix
\crefalias{section}{appendix}
\crefalias{subsection}{appendix}
\crefalias{subsubsection}{appendix}

\section{The clock-based kernel construction for \texorpdfstring{$\QMA_1$}{QMA1}} 
\label{sec:clock-construction}

We recap the matrix $K'$ whose kernel certifies \emph{yes} instances in our first version~\cite{LV26}, with $K'$ specified in \Cref{eq:finite-clock-kernel} and the corresponding witness state given in \Cref{eq:finite-clock-witness}.
Assume that $\norm{M}\geq3/4$ for \emph{yes} instances and $M\preceq I/4$ for \emph{no} instances.
Let $\sfC$ be the $\lceil\log(T+1)\rceil$-qubit clock register, where $T=w+3$. Let $L\coloneqq\sum_{j=0}^{T-1}\ketbra{j+1}{j}$ be the lower-shift matrix. 
Let $h$ denote the number of Hadamard gates in the verification circuit.
Define $\ket{b_z} \coloneqq \ket{\bar{0}}_{\sfC}\otimes\ket{z}_{\sfM}$ for $z\in\binset^w$. 
We are now ready to define the matrix $K'$: 
\begin{equation}
    \label{eq:finite-clock-kernel}
    K'\coloneqq\alpha'BB^\dagger-\beta'\ketbra{b_z}{b_z},
    \quad\text{where } B \coloneqq I-2L\otimes M.
\end{equation}
Here, $\alpha'\coloneqq m/2^k$ and $\beta'\coloneqq 2^{2hT-k}$, with $m\in\bbZ$ satisfying $2^{2hT+3}\leq m<2^{2hT+2T+2}$ and $k$ denoting its bit length. One can verify that $1/2 \leq \alpha' <1$ and $0 < \beta' \leq 1/16$. 

The linear system $B\ket{\varphi_z} = \ket{b_z}$ is a monomial instance of the matrix-generating-function framework in~\cite[Section 3]{LS26}. A direct calculation gives $B^{-1}=\sum_{j=0}^{T}L^j\otimes(2M)^j$, where the series truncates exactly because $L^{T+1}=0$. Hence, the (unnormalized) forward history vector and its squared norm are
\begin{equation}
    \label{eq:forward-vector}
    \ket{\varphi_z}=B^{-1}\ket{b_z}
    =\sum_{j=0}^{T}\ket{j}_{\sfC}\otimes(2M)^j\ket{z}_{\sfM}
    \quad\text{and}\quad
    \vartheta_z=\norm*{\ket{\varphi_z}}^2.
\end{equation}

Combining \Cref{eq:finite-clock-kernel,eq:forward-vector} gives the factorization of $K'$: 
\[
 K'=\alpha'BB^\dagger-\beta'\ketbra{b_z}{b_z}
 =\alpha'B\rbra*{I-\frac{\beta'}{\alpha'}
   \ketbra{\varphi_z}{\varphi_z}}B^\dagger.
\]

This factorization admits a \emph{geometric} interpretation: the quadratic form of $BB^\dagger$ measures squared length after applying $B^\dagger$. 
For the honest choice $\alpha'/\beta'=\vartheta_z$, the middle factor (a rank-one perturbation of $I$) is the \emph{orthogonal} projector onto $\ket{\varphi_z}^{\perp}$: it removes the forward-history direction in these coordinates.
The witness state is the normalized adjoint solution
\begin{equation}
    \label{eq:finite-clock-witness}
    \ket{\psi'_z} \coloneqq \frac{\ket{\phi_z}}{\norm*{\ket{\phi_z}}}, 
    \quad\text{where }
    \ket{\phi_z}=(BB^\dagger)^{-1}\ket{b_z} \text{ and } B^\dagger\ket{\phi_z}=\ket{\varphi_z}.
\end{equation}
For \emph{yes} instances, choosing $z$ to maximize $\vartheta_z$
makes $m=2^{2hT}\vartheta_z$ an integer in the required range.
The resulting verifier tests this kernel using an exact block encoding
of $K'/16$.

\end{document}